\documentclass[11pt]{article}
\usepackage{fullpage}

\usepackage{amsthm}
\usepackage{amsfonts,amsmath}
\usepackage{amssymb}
\usepackage{paralist}
\usepackage{multirow}
\usepackage{xspace}
\usepackage{graphicx}
\usepackage{tikz}
\usepackage{lipsum}                     
\usepackage{xargs}                      
\usepackage[prependcaption,textsize=scriptsize]{todonotes}
\usepackage[ruled,vlined]{algorithm2e}
\usepackage{colortbl}
\usepackage{lscape}
\usepackage{comment}
\usepackage{framed}
\usepackage{float}
\usepackage{tikz}
\usepackage{mathtools}
\usepackage{tabularx}
\usetikzlibrary{positioning}

\definecolor{Darkblue}{rgb}{0,0,0.4}
\definecolor{Brown}{cmyk}{0,0.61,1.,0.60}
\definecolor{Purple}{cmyk}{0.45,0.86,0,0}
\usepackage[colorlinks,linkcolor=Darkblue,filecolor=blue,citecolor=blue,urlcolor=Darkblue,pagebackref]{hyperref}
\usepackage[nameinlink]{cleveref}

\theoremstyle{theorem}
\newtheorem{theorem}{Theorem}
\newtheorem{lemma}{Lemma}

\newtheorem{claim}{Claim}
\newtheorem{corollary}{Corollary}
\newtheorem{fact}{Fact}

\theoremstyle{definition}
\newtheorem{remark}{Remark}
\newtheorem{definition}{Definition}
\newtheorem{example}{Example}

\newcommand{\NSM}{\textsf{NSM}\xspace}
\newcommand{\NSMm}{\textsf{NSM}_{2k,m,\eps}\xspace}
\newcommand{\NSMf}{\textsf{NSM}_{2k,m,\frac{1}{4}}\xspace}
\newcommand{\eps}{\epsilon}

\newcommand{\E}{{\mathbb{E}}}

\newcommand{\pref}{\ensuremath{\succ}}
\newcommand{\calR}{\mathcal{R}}

\newcommand{\Sc}{\mathrm{Sc}}

\newcommand{\poly}{\mathrm{poly}}

\newcommand{\mypara}[1]{\medskip\noindent\textbf{#1}.\quad}
\definecolor{forestgreen}{rgb}{0.1333, 0.5451, 0.1333}
\definecolor{lovelyorange}{rgb}{0.7333, 0.1451, 0.5333}

\newcommandx{\atodo}[2][1=]{\todo[linecolor=red,backgroundcolor=green!25,bordercolor=red,#1]{AF: #2}}
\newcommandx{\note}{\todo[linecolor=red,backgroundcolor=red!25,bordercolor=red]{AF}}
\newcommandx{\ado}[2][1=]{\todo[linecolor=red,backgroundcolor=blue!25,bordercolor=red,#1]{ToDo: #2}}
\newcommandx{\donow}[1]{\ado[inline]{#1}}
\newcommandx{\bdo}[2][1=]{\todo[linecolor=red,backgroundcolor=red!25,bordercolor=blue,#1]{ToDo: #2}}
\newcommandx{\cNT}[1]{\bdo[inline]{#1}}

\newcommand{\est}{\mathrm{est}}
\newcommand{\modd}[1]{(#1~\mathrm{mod}2)}

\usepackage{lineno}

\begin{document}

\title{Distributed Monitoring of Election Winners\thanks{%
  A preliminary version of this paper was presented at the
  16th International Conference on Autonomous Agents and Multiagent Systems (AAMAS~'17)~\cite{FiltserT17}.
  This full version contains all proofs, has improved upper bounds, considers more voting rules,
  studies further lower bounds, and discusses several issues in more detail.}}
\author{Arnold Filtser\thanks{%
		Ben-Gurion University of the Negev. Email: \texttt{arnoldf@cs.bgu.ac.il}. The research was supported in part by ISF grant (1718/18) and by BSF grant 2015813.
	}
	\and
	Nimrod Talmon\thanks{%
		Ben-Gurion University of the Negev. Email: \texttt{talmonn@bgu.ac.il}}
}
\maketitle

\begin{abstract}
We consider distributed elections,
where there is a center and $k$ sites.
In such distributed elections,
each voter has preferences over some set of candidates,
and each voter is assigned to exactly one site such that each site is aware only of the voters assigned to it.
The center is able to directly communicate with all sites.
We are interested in designing communication-efficient protocols,
allowing the center to maintain a candidate which,
with arbitrarily high probability,
is guaranteed to be a winner,
or at least close to being a winner.
We consider various single-winner voting rules,
such as variants of Approval voting and scoring rules,
tournament-based voting rules,
and several round-based voting rules.
For the voting rules we consider,
we show that,
using communication which is logarithmic in the number of voters,
it is possible for the center to maintain such approximate winners;
that is, upon a query at any time the center can immediately return a candidate which is guaranteed
to be an approximate winner with high probability.
We complement our protocols with lower bounds.
Our results are theoretical in nature and relate to various scenarios,
such as aggregating customer preferences in online shopping websites or supermarket chains
and collecting votes from different polling stations of political elections.
\end{abstract}

\section{Introduction}

Elections are extensively used to aggregate preferences of voters.
Some elections are centralized,
but others are carried out in distributed settings.
Consider,
for example,
a supermarket chain consisting of a large number of stores.
Each store collects data on the purchases made in it,
and the managers at the chain headquarters might want to aggregate this data,
to identify,
for example,
the most popular items being sold.
One solution would be to have a central database,
collecting all data from all stores,
and to compute the most popular items on this centralized database.
As the number of customers might be huge,
however,
it might not be practical to do so.
Further,
as the communication between the stores and the headquarters might be expensive,
a more efficient solution would be to have some computations being made locally at each store,
and to develop a protocol for efficient communication between the stores and the headquarters,
to allow the managers at the headquarters to know,
at each point in time,
what are the most popular items that are being sold throughout the chain.

A similar situation happens in online shopping websites,
where buyers from all around the world make purchases.
As the design of modern websites is based on data centers,
aggregating the data concerning all buyers involves communicating in a distributed setting.
Specifically,
in order to identify the current trends,
and as communication between data centers might be expensive,
it is of interest to develop protocols for those data centers to communicate with a central entity.

Our model also captures scenarios of political polls and political elections.
That is,
in political elections and in TV polls, it is usually the case that there are several polling stations,
spread around the country or the region.
Then,
in order to compute the results of the election (or the intermediate results during the day when the poll is being held),
the voters' preferences from all those polling stations are aggregated at some central station.
For example,
in the general political elections held in Brazil in 2014,
there were roughly 500,000 polling stations,
with an average of 300 voters per station.
In this situation,
it is beneficial to have a protocol allowing the polling stations to efficiently communicate
with a central entity, allowing the central entity to maintain a good estimate on the nation-wide (or region-wide) state of affairs.

In this paper,
we model such situations as follows.
We are considering an election whose electorate is distributed into $k$ sites.
Assuming some common axis of time\footnote{
  To avoid confusion,
  let us mention that,
  while we indeed speak about ``time'',
  we do not consider any external clocks (or, importantly, clocks accessible to the sites or the center).
  In particular, the voters can be assumed to come at fixed intervals,
  whose speed is not known to the sites nor to the center.
},
we have that at each point in time,
a new voter arrives and votes,
and her vote is assigned to one of those $k$ sites\footnote{For convenience, we refer to voters as females,
while the candidates are males.}.
There is some center which is able to directly communicate with each of the $k$ sites.
With respect to a voting rule~$\calR$,
the goal of the center is to maintain,
at any point in time,
a candidate which is an $\calR$-winner of the whole election
(given an election $E$ and a voting rule $\calR$, an $\calR$-winner of $E$ is a winner of $E$ under $\calR$).
More specifically,
we are interested in designing communication-efficient protocols,
where the center is able, upon request at any time, to return a candidate which,
with high probability,
is an $\calR$-winner.

As we are interested in sublinear communication, 
in addition to allowing mistakes to accrue with some low probability, we will also use approximation.
We call a candidate an $\eps$-winner with respect to a voting rule~$\calR$, if by adding up to an $\epsilon$-fraction of new voters, it can become an $\calR$-winner.
A more formal description of our model and a discussion on our notion of approximation is given in \Cref{section:preliminaries}.
Previous works were concerned with bribery (where we are allowed to change an $\epsilon$-fraction of the voters), and margin of victory (where we are guaranteed that by changing an $\epsilon$-fraction of the voters, the outcome of the election shall remain unchanged), see \Cref{sec:related} for additional details on these notions.
These notions are appropriate to deal with noisy data, or to be used in scenarios where some external agent can influence the voters, thus change their votes.
Here, however, we are concerned with monitoring an election while minimizing the communication, and the source of our errors is lack of information (rather than noise). Our approximation notion fits better to our scenario, as a candidate is an $\eps$-winner if it might become a winner under full information.
Furthermore, in monitoring an election we do expect more voters to come, thus, in this aspect, an $\eps$-winner is a candidate who might become a winner very shortly.
Finally, as we consider an ongoing election, changing previous votes is not an option. However, the information on whether a candidate is an $\eps$-winner is very valuable for making, e.g., real-time election policy decisions.

We concentrate on single-winner voting rules,
and consider various voting rules,
ranging from approval-based rules and scoring rules,
to tournament-based rules and round-based voting rules;
while we naturally cannot cover all voting rules available,
we choose some of the more popular and representative ones as well as aim at choosing representative voting rules.
Further,
we develop some general techniques for designing protocols for maintaining approximate winners in distributed elections,
which might be applicable to other voting rules and settings as well.
We show how to apply these techniques for the rules we consider.
We discuss the effect of several parameters on the communication complexity of the protocols we design;
specifically,
the effect that the number $n$ of voters, the number $m$ of candidates, the required approximation $\epsilon$, and the number $k$ of sites
have on the amount of communication used by our protocols.
We complement our communication-efficient protocols with lower bounds.

As a by-product of our lower bounds for maintaining an approximate Plurality winner in distributed elections,
we have two contributions which might be useful in other contexts.
First,
we improve the state-of-the-art lower bound on the \textsc{Count-tracking} problem,
which is a central problem in distributed streams; this result is discussed in detail in \Cref{remark:byproduct}.
In short,
in the \textsc{Count-tracking} problem,
the task is to maintain a value which approximates the number of items in a given distributed stream.
In the regime where $k \geq 1/\epsilon^2$, we improve the lower bound for \textsc{Count-tracking} from $\Omega(k)$,
proved by Huang et al.~\cite[Theorem~2.3]{huang2012randomized},
to $\Omega(k \log n / \log k)$ (see \Cref{remark:byproduct}).
Second,
we define a novel problem in multiparty communication complexity and show a tight lower bound for it;
in this problem,
which we call the \emph{No Strict Majority} problem,
we have several players, each with its own private binary string,
and,
by communicating bits,
the players should decide whether there is some index for which a majority of the players have $1$ in it.
We prove a lower bound on the \emph{No Strict Majority} problem,
showing that the naive protocol for this problem is essentially optimal:
  asymptotically, all the bits have to be transmitted.
See \Cref{section:lowerbounds} for further details on our lower bounds and their implications to
continuous distributed monitoring and to multiparty communication complexity.

\subsection{Related Work}\label{sec:related}

We first review related work on sublinear algorithms in computational social choice,
as the current paper fits naturally within this line of research.
Then we review papers on compilation complexity,
vote elicitation,
and mention some connections between our notion of approximation to work on control and bribery in elections
(as well as to the concept of margin of victory).
Finally,
we give an overview on the available literature on the continuous distributed monitoring model,
which is the computational model we use in the current paper
(its formal definition is given in \Cref{section:preliminaries}).

\mypara{Sublinear social choice}
As the amount of data in general, and data concerning preferences in particular,
is consistently increasing,
the study of identifying election winners using time or space which is sublinear in the number of voters
is receiving increasing attention.
Specifically, the size of some elections might be too big to process in linear time,
thus algorithms with sublinear time and/or space complexity are of interest.

We first mention a follow-up to the conference version of this paper which was recently published~\cite{chaturvedi2019distributed}.
That paper provides a different protocol for winner tracking which is based on our checkpoints technique.
Their technique achieves an improvement of a $(1 + \frac{\log k}{\log n/k})$ over our checkpoints-based protocols.
Notably, they also performed computer simulations to evaluate the practical communication complexity of theirs and our protocols.

In two papers,
Bhattacharyya and Dey~\cite{dey2015sample, bhattacharyya2015fishing}
study sampling algorithms for winner determination as well as winner determination in the streaming model.
In fact, some of our sampling-based protocols are inspired by Bhattacharyya and Dey~\cite{dey2015sample}.
In their model, they assume that they are given an election in which the margin of victory is at least $\epsilon n$
(where $n$ is the number of voters);
  this means that the winner is guaranteed to remain such even if an adversary is allowed to change $\epsilon n$ votes.
Given such elections, they evaluate the number of vote samples needed in order to identify the winner with high probability.
In our current paper,
we have a different notion of approximation and we do not assume such margins of victory
(we formally describe our notion of approximation in \Cref{section:preliminaries}).

\begin{remark}
There is a mistake in the preliminary version of this work~\cite{FiltserT17},
which claims that the sampling-based protocols are implied by the work of Bhattacharyya and Dey~\cite{dey2015sample, bhattacharyya2015fishing}.
This is incorrect as our notion of approximation is different than theirs, specifically due to this margin of victory assumption which in particular means that, while an approximate winner under our definition always exists, this does not necessarily hold in their model.
\end{remark}

In a recent paper,
Dey et al.~\cite{DTH17}
study winner determination for several multiwinner voting rules aiming at proportional representation.
Dey and Narahari~\cite{dey2015estimating} study sampling algorithms for estimating the margin of victory.
These works deal with centralized elections, while the current paper considers distributed elections.
Another paper worth mentioning in this context is the paper of Lee et al.~\cite{lee2014crowdsourcing} which argues for the importance of developing fast communication-efficient protocols for computing winner in (centralized) streams; they also provide a simple sampling-based algorithm for approximating Borda winners.

Not strictly considering sublinear social choice,
but nonetheless concentrating on ``huge elections'',
in a recent paper, Csar et al.~\cite{CLPS17}
  study winner determination using the MapReduce framework
which may allow processing such elections efficiently by distributing the computation among clusters of machines.

\mypara{Compilation complexity}
In a series of papers,
Chevaleyre et al.~\cite{chevaleyre2009compiling, chevaleyre2011compilation}
and Xia and Conitzer~\cite{xia2010compilation} define and study the compilation complexity of various voting rules;
in their model, the electorate is partitioned into two parts,
and the general concern is the amount of communication which needs to be transmitted
between the two parts, in order to determine an election winner.
In compilation complexity there are no rounds of communication,
as only one message is being passed between the two parts.
This stands in contrast to our protocols,
which use small amounts of communication due to their use of several rounds of communication
between the center and the sites.

\mypara{Vote elicitation}
There is quite an extensive literature which deal with vote elicitation~\cite{dhamal2013scalable, con-san:c:strategy-proofness, lee2014crowdsourcing, lee2015efficient};
these works provide algorithms for finding approximate winners under various voting rules,
by elicitating the voters' preference orders.
Conitzer and Sandholm~\cite{conitzer2005communication} study communication complexity for various voting rules. In their model, each voter acts as a site.
Conitzer and Sandholm were interested in finding exact winners. 
In a follow up paper, Service and Adams~\cite{SA12} allowed approximation. For example, for Borda, they define the approximation ratio of a candidate $c$ to be $\frac{{\rm sc}(c)}{{\rm sc}(w)}$, where $w$ is the true winner and ${\rm sc}$ is the Borda score function. Their goal is to find a candidate with $1-\eps$ approximation ratio; approximation ratios for other voting rules are defined similarly.
However, as each voter acts as a site, their upper bounds are quite high:
  In particular, they depend linearly on the number of voters.
It is also interesting to mention that our notion of approximation is strictly stronger than theirs, at least for Borda (that is, an $\eps$-Borda winner has $O(\eps)$ approximation ratio under the definition used by Service and Adams~\cite{SA12}).

\mypara{Approximate winners, margin of victory, and election control}
In the current paper we do not require our protocols to maintain exact winners,
but are satisfied with approximate winners.
We formally define our notion of approximation in \Cref{section:preliminaries};
roughly speaking,
we consider a candidate to be an approximate winner if it can become a winner if we are allowed
to add a small number of additional voters (where we can set their votes as we wish).
Our notion of approximation somehow resembles the vast amount of research done on electoral control and bribery in elections
(see, e.g., the survey by Faliszewski and Rothe~\cite{fal-rot:b:control}).
In electoral control by adding voters, there is usually a set of unregistered voters,
and the question is whether it is possible to change the outcome of the election,
e.g., to have some predefined, preferred candidate to become a winner in a new election,
where a small number of those unregistered voters are added to the election.

In bribery problems, such as shift bribery and swap bribery~\cite{elk-fal-sli:c:swap-bribery},
an external agent can change the way some voters vote in order to have some predefined, preferred candidate to become a winner.
As observed by Xia~\cite{xia:margin-of-victory},
the number of such changes that needs to be done in order to make a specific candidate to become a winner
(the so-called margin of victory),
is a natural notion of this candidate's closeness to be a winner.
Indeed, in this sense, our approximation notion is related to those notions of control and bribery in elections.

\mypara{Continuous distributed monitoring}
The model of computation which we study in the current paper is called
the \emph{continuous distributed monitoring} model,
and is usually studied within theoretical computer science and database systems.
There is a fairly recent survey about this model~\cite{cormode2013continuous},
as well as quite extensive line of work studying various problems in this model,
such as sampling-based protocols~\cite{cormode2012continuous, tirthapura2011optimal},
protocols for approximating moments~\cite{cormode2011algorithms,arackaparambil2009functional},
protocols for counting with deletions~\cite{liucontinuous}
(interestingly, that paper specifically mentions elections as a motivation, but do not study it explicitly),
heuristic protocols for monitoring most-frequent items~\cite{babcock2003distributed},
and randomized protocols for counting the number of items in a distributed stream
and finding frequent items~\cite{huang2012randomized}.
In the current paper we complement this line of work by studying winner determination in this model.

\section{Preliminaries}\label{section:preliminaries}

We begin by providing preliminaries regarding elections and voting rules,
continue by describing our notion of approximation,
and finish by discussing our model concerning continuous monitoring of distributed streams.
We use standard notions from computational complexity.
For $n \in \mathbb{N}$,
we denote the set $\{1, \ldots, n\}$ by $[n]$.

\subsection{Elections and Voting Rules}\label{subsec:PrelimVotingRules}

An \emph{election} $E = (C, V)$
consists of a set of \emph{candidates} $C = \{c_1, \ldots , c_m\}$
and a collection of \emph{voters} $V = (v_1, \ldots , v_n)$.
We consider both \emph{approval} elections,
where voters cast approval ballots,
and \emph{ordinal} elections,
where voters cast ordinal ballots.

Specifically,
in approval elections,
each voter is associated with her set of approved candidates,
such that $v_i \subseteq C$.
We say that $v_i$ \emph{approves} candidate $c$ if $c \in v_i$
(and \emph{disapproves} him otherwise).
In ordinal elections each voter is a total order $\pref_{v_i}$ over $C$.
A \emph{single-winner voting rule} $\calR$ is a function that gets an election $E$
and returns a set $\calR(E) \subseteq C$ of co-winners of that elections,
such that $c$ is a winner of the election $E$ under $\calR$ if $c \in \calR(E)$.

Next we define our voting rules of interest.
We ignore issues of tie-breaking;
specifically,
we assume an arbitrary tie-breaking order which works in our favor,
such that a candidate $c$ is a winner if there is some fixed tie-breaking that makes him a winner.

We begin with approval-based voting rules and scoring rules,
continue with tournament-based voting rules,
and then discuss round-based voting rules.

\subsubsection{Approval-based Rules and Scoring Rules}

\mypara{Plurality, $t$-Approval, and Approval}
Under \emph{Approval},
each voter approves a subset of the candidates (that is, it is held in approval elections),
and the score of a candidate is the number of voters approving him.
The candidates with the highest score tie as co-winners.
\emph{$t$-Approval} is similar to Approval,
but with the restriction that each voter shall approve exactly $t$ candidates
(that is, $|v_i| = t$; we assume that $t \leq m / 2$).
\emph{Plurality} is a synonym for $1$-Approval,
that is, where each voter approves exactly one candidate.

\mypara{Borda}
Borda is the archetypical scoring rule.
Under \emph{Borda}, a voter ranking a candidate at position $j$ is giving him $m - j$ points,
and the candidates with the highest score tie as co-winners.
Scoring rules in general are defined similarly, but with scoring vectors other than the one used by Borda.

\subsubsection{Tournament-based Voting Rules}

\mypara{Cup}
The \emph{Cup} voting rule is defined via a balanced binary tree $T$ with $m$ leaves,
such that there is exactly one leaf for each candidate.
Starting from the leaves,
in a bottom-up fashion,
each non-leaf node is associated with the candidate which wins in the pairwise election
held with only the two candidates corresponding to the two children of that node.
Finally,
the candidate which gets assigned to the root of $T$ is declared the winner of the election.

\mypara{Copeland and Condorcet}
The \emph{Copeland score} of a candidate $c$ is the number of other candidates $c' \neq c$ for which a majority of voters prefer $c$ to $c'$.
Under \emph{Copeland},
the candidates with the highest Copeland score tie as co-winners.
A \emph{Condorcet winner} is a candidate with Copeland score $m - 1$.
Under \emph{Condorcet},
a Condorcet winner is selected as a winner if it exists;
otherwise, all candidates tie as co-winners.

\subsubsection{Round-Based Voting Rules}

\mypara{Plurality with run-off}
\emph{Plurality with run-off} proceeds in two rounds.
In the first round, it selects two candidates with the highest Plurality scores,
where the Plurality score of a candidate is defined as the number of voters ranking him first.
In the second round,
it considers only those two candidates selected in the first round
and selects as a winner the one which is preferred to the other by majority of voters.

\mypara{Bucklin}
\emph{Bucklin} also proceeds in rounds.
In round $i \in [m]$,
it computes, for each candidate $c$, the number of voters ranking $c$ among their top $i$ choices.
Then,
if there is a candidate with a strict majority of the voters ranking him among their top $i$ choices,
then such a candidate is selected as a winner;
otherwise, a new round begins.

\subsection{Our Notion of Approximation}

Since we will be interested in designing protocols where the center cannot see the full election,
it will not be possible to guarantee that our protocols will find exact winners;
therefore,
we will be satisfied with protocols which are guaranteed to find approximate winners.
There are several possibilities for defining approximate winners of elections;
in this paper we consider $\epsilon$-winners.
Roughly speaking,
an $\epsilon$-winner is a candidate which is not far from being the winner of the election in the sense that he might become a winner after the arrival of only a few additional new voters.
A more formal definition follows.

\begin{definition}[$\epsilon$-winner]
For $\epsilon\in(0,1)$, a candidate $c$ is an \emph{$\epsilon$-winner}
  	in an election $E$ (with $n$ voters) under some voting rule $\calR$
  	if it can become a winner under $\calR$ by adding at most $\epsilon n$ additional voters to $E$.
  	That is, if there exists an election $E'$, where $E\subseteq E'$ and $|E'\setminus E|\le \eps\cdot n$ such that $c \in \calR(E')$.
\end{definition}  

Indeed,
we view the definition of an $\epsilon$-winner as a definition of approximation,
as the lower $\epsilon$ is, the closer an $\epsilon$-winner is to a real winner.
As we will design our protocols to compute $\epsilon$-winner,
the lower $\epsilon$ would be,
their guaranteed results would become closer and closer to real winners.

Our approximation notion seems particularly relevant to our setting (as compared to, e.g., the notion used by Bhattacharyya and Dey~\cite{dey2015sample, bhattacharyya2015fishing}), for the following reasons.
First, we do not assume a margin-of-victory, namely that some candidate is a clear winner. Second, in distributed vote streams we expect more voters to arrive in the future, thus we are interested in identifying a candidate which might become a winner in the near future: such candidates are exactly the $\epsilon$-winners.\footnote{As a side note, we mention that in political elections such a knowledge might worth much to these candidates, as it can help them decide on when to spend their campaigning funds.}

\subsection{Our Model of Computation}

In our computational model, we have one center and $k$ sites.
The center and the sites are arranged in a star-shaped network,
centered at the center,
such that the center has a direct communication link to each site
but two sites cannot communicate directly.

We assume some axis of time, $t_1, \ldots, t_n$,
and a stream of voters $v_1, \ldots, v_n$,
such that voter $v_i$ comes at time $t_i$.
Each voter is assigned to exactly one site,
such that each site is aware only of the subset of voters which are assigned to it.
We stress that the time is not known to either the center or the sites.
Such a stream is called a \emph{distributed stream}.
\Cref{figure:model} illustrates the model. \Cref{table:example} provides an example of a distributed stream.

We mention that our model of computation might be seen as the model of computation assumed in the study of \emph{Continuous Distributed Monitoring},
when instantiated for vote streams (and not general, abstract streams).
See the Related Work section for more details on this subject.

\begin{figure}
\centering
\includegraphics[scale=0.4]{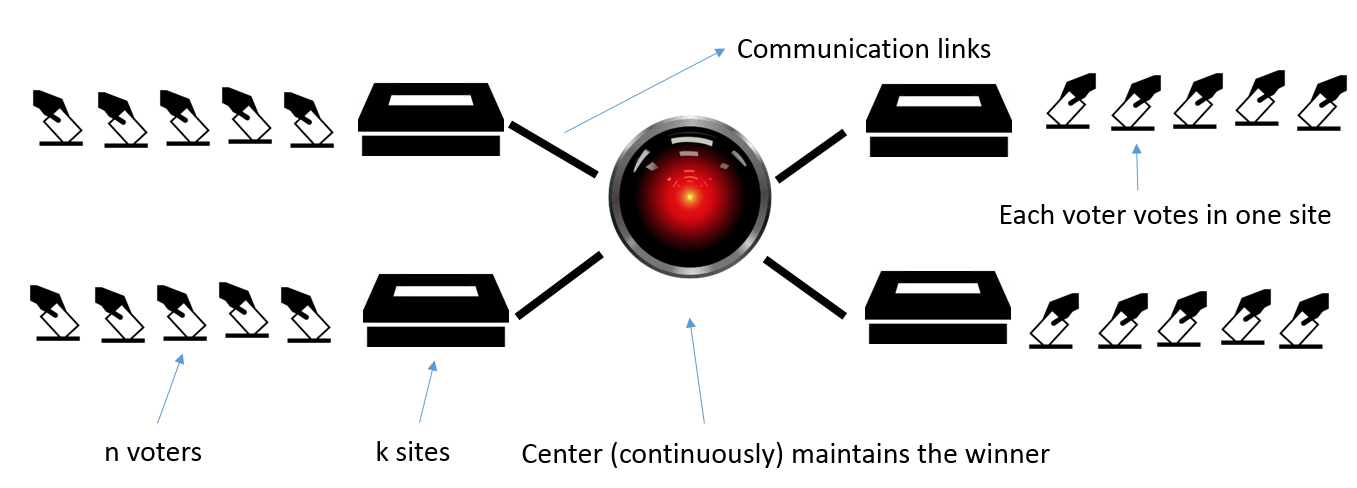}
\caption{Illustration of our model.}
\label{figure:model}
\end{figure}

\begin{table*}[th]
	\centering
	\begin{tabular}{c | c c c}
		$t$ & $S_1$ & $S_2$ & $S_3$ \\
		\hline
		$1$ & $a$ & & \\
		$2$ & $a$ & & \\
		$3$ & & $b$ & \\
		$4$ & & & $a$ \\
		$5$ & & & $a$ \\
		$6$ & & $b$ & \\
		$7$ & & $b$ &
	\end{tabular}
	\caption{An example of a distributed stream corresponding to a Plurality election over the candidate set $\{a, b\}$. Rows correspond to the exist of time, $t_1, \ldots, t_n$, while each column corresponds to a site; e.g., at time $3$, a voter voting for $b$ is assigned to site $S_2$.
	Notice that the situation at a certain time can be viewed as the current election; for example, the election at time $4$ consists of $3$ votes for $1$ and $1$ vote for $b$. For this election, notice that, while $a$ is the winner, and thus an $\epsilon$-winner for any $0 \leq \epsilon \leq 1$, $b$ is not an $\epsilon$-winner even for $\epsilon = 0.5$, as even adding $\epsilon n = 0.5 \cdot 4 = 2$ votes for $b$ would not make $b$ a Plurality winner.}
	\label{table:example}
\end{table*}

We are interested in designing communication-efficient protocols,
whose goals are to allow the center to declare,
at any point in time,
a candidate $c$ which is,
with constant probability (say, $0.9$),
an $\epsilon$-winner
(see \Cref{section:results} for a discussion on higher probabilities).
Formally, for a voting rule $\calR$, 
we are interested in finding efficient protocols for the \textsc{$\calR$-winner-tracking} problem.

A protocol is defined via the messages which the center and the sites send to each other,
and can consist of several rounds.
The protocol shall be correct not only at the end of the stream
(which is usually the case in streaming algorithms),
but shall be correct at any point in time.
As it is the custom in protocols operating on distributed streams,
we describe our upper bounds in terms of words of communication,
where we assume that a word contains $\log n$ bits.

We measure the communication complexity of a protocol with respect to the expected number of words communicated in its run. Using Markov inequality, this expectation can be replaced by high probability (or even by probability $1$ in the expense of slight increase in the failure probability).

\subsection{Useful Results from Probability Theory}

For the sampling based protocols, we will use the following bound.

\begin{theorem}[Chernoff Bound]\label{thm:chernoff}
Let $X_1, \dots, X_s$ be a sequence of $s$ i.i.d
random variables in $[0,1]$. Let $X = \sum_i X_i$ and
let $\mu = \E[X]$. Then, for any $0 \leq \gamma \leq 1$: 
$$\Pr[|X - \mu| \geq \gamma \mu] < 2\exp(-\gamma^2\mu/3)~.$$
\end{theorem}

The following corollary will be the main building block for our sampling-based protocols.

\begin{lemma}\label{lem:BasicSampler}
  Let $X_1, \dots, X_s$ be i.i.d random variables in $[0,1]$ with mean $p$.
  Let $X=\sum_i X_i$ and let $q=\frac1s X$.
  Then,
  for $s\ge \frac{3}{\eps^2}\log(\frac2\delta)$ it holds that
  $$\Pr[|q-p|\ge\eps]<\delta~.$$
\end{lemma} 

\begin{proof}
Set $\mu=\E[X]=s\cdot p$.
Using the Chernoff Bound (i.e., \Cref{thm:chernoff}),
it follows that:
\[
\Pr\left[\left|q-p\right|\ge\epsilon\right]
=\Pr\left[\left|X-\mu\right|\ge\frac{\epsilon}{p}\cdot\mu\right]\le2\exp\left(-\left(\frac{\epsilon}{p}\right)^{2}\cdot\mu/3\right)\le2\exp\left(-\epsilon^{2}s/3\right)\le\delta~.\qedhere
\]
\end{proof}

\section{Algorithmic Techniques}

The naive protocol,
where each site sends to the center a message for every voter that arrives to it,
clearly solves our problem,
however it uses communication which is linear in the number of voters.
For example,
for ordinal ballots,
it communicates $O(n \cdot m \log m)$ bits,
since $m \log m$ bits are sufficient for sending a single vote.
In this paper we are interested in protocols which use significantly less communication,
namely communication which is polylogarithmic in the number of voters.

In this section we provide high level descriptions of three algorithmic techniques
which are useful for developing protocols for maintaining approximate winners in distributed vote streams.
Accordingly,
in \Cref{section:results} we demonstrate how to realize and instantiate these algorithmic techniques
as concrete protocols for maintaining approximate winners
for various specific voting rules.

\subsection{Protocols Based on Counting Frequencies}\label{section:protocols-based-on-counting-frequencies}

In the \textsc{Frequency-tracking} problem,
we are given a distributed stream where,
instead of voters,
the items of the stream come from a known universe of items.
The goal is for the center to maintain,
for each item type in the distributed stream,
a value which approximates the frequency of that item type.
More formally,
let us denote the items of the stream by $v_1, \ldots, v_n$
and consider $m$ different item types,
such that item $i\in[n]$ is of type $j\in [m]$
if $v_i = j$.
Let us denote the frequency of item type $j$ by $f(j) = |\{i : v_i = j\}|$.
A protocol solving the \textsc{Frequency-tracking} problem guarantees that, with constant probability,
simultaneously for every item type $j$,
the center can maintain a value $f'(j)$ such that $f'(j) \in f(j) \pm \epsilon n$.

Estimating the frequencies of item types is a fundamental problem in distributed streams
(in fact, also in centralized streams).
A deterministic protocol for \textsc{Frequency-tracking},
using $O(\epsilon^{-1} k \log n)$ words of communication, is known~\cite{yi2013optimal}; in fact, it is also known to be tight.
Moreover, 
there is a randomized protocol which uses
$O((\epsilon^{-1} \sqrt{k} + k) \log n\log \frac1\delta)$ words of communication, and succeeds with probability $1-\delta$ ~\cite{huang2012randomized}.\footnote{%
  Notice that Huang et al.~\cite{huang2012randomized} consider only situations where $k \leq \epsilon^{-2}$,
  thus their bounds read differently; nevertheless, $O((\epsilon^{-1} \sqrt{k} + k) \log n\cdot\log \frac1\delta)$ is the communication complexity of their protocol.}
Formally, the protocol guarantees that for every $j\in[m]$ and every $n$, after the arrival of $n$ items, $\Pr[f'(j) \in f(j) \pm \epsilon n]\ge 1-\delta$.
In particular, by setting $\delta=1/\poly(m)$ and applying the union bound, we get that for every $n$,  $\Pr[\forall j,~~~f'(j) \in f(j) \pm \epsilon n]\ge 1-\frac{1}{\poly(m)}$. The communication complexity in this case is $O((\epsilon^{-1} \sqrt{k} + k) \log n\cdot\log m)$.

Many voting rules operate by counting points for candidates,
thus,
it can be seen as if these voting rules actually count frequencies of, say, approvals of each candidate.
It turns out that,
indeed,
it is sometimes possible to reduce the problem of maintaining an $\epsilon$-winner
under such voting rules to the problem of maintaining approximate frequencies.
For more concrete intuition, consider the following distributed stream for the \textsc{Frequency-tracking} problem.

\begin{example}
Consider the set of $m = 3$ item types $\{1, 2, 3\}$ and a distributed stream containing the following $n = 10$ items: $[2, 2, 1, 3, 3, 3, 3, 3, 2, 1]$. Then, the frequencies $f$ of the item types is: $f(1) = 1$, $f(2) = 3$, and $f(3) = 5$.
Observe how the above stream might, roughly speaking, correspond to a distributed vote stream for \textsc{Plurality-winner-tracking} containing $10$ votes over the set of candidates $\{1, 2, 3\}$, and how, in essence, Plurality simply counts the frequencies of the corresponding items.
\end{example}

During the description of our results for specific voting rules, in \Cref{section:results}, we will usually use the randomized version of the \textsc{Frequency-tracking} protocol, the only exception being the hybrid protocol for Runoff, for which we will use the deterministic version.

\subsection{Protocols Based on Checkpoints}\label{section:protocols-based-on-checkpoints}

Protocols based on checkpoints are deterministic in nature,
and the general idea behind such protocols is as follows.
Assume that the center knows an $\epsilon$-winner $c$ of the election containing the first $n$ voters.
Then,
the crucial observation is that,
until the number of voters reaches $(1 + \epsilon )n$, 
the center can declare $c$ as an $O(\epsilon)$-winner.
This suggests protocols where the center only updates its declared candidate
whenever the number of voters multiplies by a $(1 + \Omega(\eps))$-fraction.
Such points of time will be called \emph{checkpoints}.
Between two checkpoints, the center will declare the  previous estimation as the current $\eps$-winner.
This intuition is formulated in the following lemma, the proof of which appears in \Cref{appendix:proofOfCheckPoints}.\footnote{%
  While some of the ideas in the proof might fit naturally in the main text,
  the proof considers each voting rule studied in this paper separately,
  and thus it is slightly repetitive, and thus deferred to the appendix.} Algorithms~\ref{algorithm:checkpointsCenter} and \ref{algorithm:checkpointsCite} provide pseudocode for a general protocol based on checkpoints.
  
\begin{algorithm}[t]
	\ \\
	Initiate and maintain an approximate count $n'$ using a \textsc{Count-tracking} protocol with $\lambda = \epsilon / 12$\; 
	\If{$n' \geq (1 + \lambda)^i$ for some $i$, for the first time for this $i$}{Initiate a static subprotocol to identify an $\frac{\epsilon}{4}$-winner $c$\;}
	\textbf{Upon query:} declare $c$\; 
	\caption{Protocol based on checkpoints, for center.}
	\label{algorithm:checkpointsCenter}
\end{algorithm}
\begin{algorithm}[t]
\ \\
Participate in the \textsc{Count-tracking} protocol initiated by the center\;
Participate in each static subprotocol initiated by the center\;

\caption{Protocol based on checkpoints, for site.}
\label{algorithm:checkpointsCite}
\end{algorithm}

\begin{lemma}\label{lem:Checkpoints}
	Let $\calR$ be some voting rule described in \Cref{subsec:PrelimVotingRules}.
	Let $E=\{v_1,\dots,v_n\}$ and $E'=E\cup\{v_{n+1},\dots,v_{n+q}\}$, where $q\le \frac{\eps}{4}n$, be two elections.
	If candidate $c$ is an $\frac{\eps}{4}$-winner w.r.t $E$, then $c$ is an $\eps$-winner w.r.t $E'$. 
\end{lemma}

In order to identify the checkpoints,
the center shall be able to count the number of voters arriving so far.
Fortunately, for an approximation factor $\lambda\in(0,1)$,
there is an efficient deterministic protocol for solving the \textsc{Count-tracking} problem,
which uses $O(\lambda^{-1} k \log n)$ words~\cite{yi2013optimal};
in the \textsc{Count-tracking} problem,
the center maintains a value $n'$ such that $n' \in n \pm \lambda n$,
where $n$ is the actual number of items in the distributed stream.

Now we have all the ingredients for our generic protocol.
Specifically,
the center will maintain a value $n'$ using a \textsc{Count-tracking} protocol with precession parameter $\lambda=\frac{\eps}{12}$. Each time $n'$ exceeds $(1+\lambda)^i$ for the first time,\footnote{In fact, the \textsc{Count-tracking} protocol of ~\cite{yi2013optimal} only increases its estimation as time goes by.} for some $i$, the center will initiate a \emph{static} subprotocol to identify an $\frac{\eps}{4}$-winner $c$ of the election so far.
The center will declare $c$ as $\eps$-winner until the next checkpoint. 
We argue that $c$ is indeed an $\eps$-winner.
Consider a step in time $n$. Then the center's estimation $n'$ of the number of voters is at least $(1-\lambda)n$. In particular, it necessarily had a ``checkpoint'' at time $n''$, for $n''\ge \frac{1-\lambda}{1+\lambda}n$. Thus $n\le(1+3\lambda)n''=(1+\frac{\epsilon}{4})n''$. By \Cref{lem:Checkpoints}, as  $c$ was $\frac{\eps}{4}$-winner at time $n''$, it is also $\eps$-winner at time $n$. 

As the estimation $n'$ is upper bounded by $(1+\lambda)n$,  the number of checkpoints is bounded by $\log_{1+\lambda}\left((1+\lambda)n\right)=O(\log n/\lambda)=O(\log n/\epsilon)$. 
Assuming that it is possible to compute an $\frac\epsilon4$-winner using $O(z)$ words,
a protocol based on checkpoints would then need $O((k +z)\epsilon^{-1} \log n)$
words of communication.
As $z$ will be at least $\Omega(k)$, we would get $O(z \cdot \epsilon^{-1} \log n)$.\footnote{%
	Huang et al.~\cite{huang2012randomized} provide a randomized protocol for \textsc{Count-tracking} which uses $O(\sqrt{k}\epsilon^{-1} \log n)$ bits of communication.
	As $z$ will be greater than $\Omega(k)$, using randomization will not reduce the total asymptotic communication.} We conclude with the following:

\begin{corollary}
	Let $\calR$ be a voting rule from those described in \Cref{subsec:PrelimVotingRules}.
	Suppose that there is a static deterministic protocol that computes an $\frac\eps4$-winner using $z$ words of communication, for $z\ge k$. Then, there is a protocol for $\calR$-winner-tracking which uses $O(z \cdot \epsilon^{-1} \log n)$ words.
\end{corollary}

During the description of our results for specific voting rules, in \Cref{section:results},
we will describe only the static protocol in each protocol based on checkpoints.
For simplicity of presentation, we will compute $\eps$-winner instead of $\frac{\eps}{4}$-winner as actually needed.

\subsection{Protocols Based on Sampling}\label{section:protocols-based-on-sampling}

Instead of sending all voters to the center,
as the naive protocol does,
it is natural to let each site send only some of the voters arriving to it.
Specifically,
we would like the center to have a uniform sample of the voters.
Cormode et al.~\cite{cormode2012continuous} describe a protocol for maintaining a sample of $s$ items chosen uniformly at random from a distributed stream; its communication complexity is $O((k + s) \log n)$.
Since we are sampling voters,
we need to take into account the communication needed to send each of the sampled voters.
Specifically,
in approval elections (where the voters cast approval ballots),
we need $m$ bits per voter.
Since we count the communication complexity in words, each of which contains $\log n$ bits,
we need $\lceil\log 2^m / \log n \rceil\le 1 + m / \log n$ words per voter.
Similarly,
in ordinal elections (where the voters cast ordinal ballots),
we need ($\log m!$) bits per voter, thus $\lceil\log m! / \log n \rceil\le 1 + m \log m / \log n$ words per voter.

But how many samples are needed in order to determine an $\eps$-winner with high probability?
Our main building block will be \Cref{lem:BasicSampler} (see \Cref{section:preliminaries}) and our general framework is as follows.
For each voting rule,
we will use \Cref{lem:BasicSampler} to argue that, with $s$ samples,
chosen uniformly with repetitions,
we can determine an $\eps$-winner with high probability.
Then, assuming that we need $w$ words of communication for each voter,
using an efficient sampling protocol~\cite{cormode2012continuous},
as discussed above,
we will get a communication protocol with complexity $O((k+s)w\cdot\log n)$.
(As we use asymptotic analysis, it will be enough to find an $O(\eps)$-winner and to adjust the parameters accordingly.)

\section{Communication-efficient Protocols}\label{section:results}

Our upper bounds are summarized in \Cref{table:results}.
We begin with approval-based rules and scoring rules,
continue with tournament-based rules,
and then discuss round-based rules.
Before we present our specific upper bounds,
the following remark, concerning the success probability of our protocols, is in place.

\begin{table*}[th]
	\centering
	\resizebox{\textwidth}{!}{\begin{tabular}{ l l l l }
			\\\hline
			\multicolumn{4}{c}{\textbf{ Upper Bounds} \phantom{$2^{2^{2^2}}$}}      \\\hline \ \\
			\textbf{Voting rule}&\textbf{Frequencies}&\textbf{Checkpoints}&\textbf{Sampling}\\\hline
			
			\textbf{Plurality}               & $O((\epsilon^{-1} \sqrt{k} + k) \log n{\cdot\log m})$ & & \\
			\textbf{$t$-Approval}            & $O((\epsilon^{-1} \sqrt{k} t + k) \log tn{\cdot\log m})$ & $O\left(\frac{k}{\epsilon}(m\log\frac{k}{\epsilon}+\log n)\right)$ & $O(\epsilon^{-2} \log (2t) + k) (\log{m\choose t} + \log n)$ \\
			\textbf{Approval}                & $O((\epsilon^{-1} \sqrt{k} m + k) \log mn{\cdot\log m})$ & $O\left(\frac{k}{\epsilon}(m\log\frac{k}{\epsilon}+\log n)\right)$ & $O((\epsilon^{-2} \log m + k) (m + \log n))$ \\
			\textbf{Borda}                   & $O((\epsilon^{-1} \sqrt{k} m + k) \log m n{\cdot\log m})$ & $O\left(\frac{k}{\epsilon}(m\log\frac{k}{\epsilon}+\log n)\right)$ &
			$O((\epsilon^{-2} \log m + k) (m \log m + \log n))$ \\
			\textbf{Condorcet}               & $O((\epsilon^{-1} \sqrt{k} m^2 + k) \log m n{\cdot\log m})$ & 	$O\left(\frac{k}{\epsilon}(m\log\frac{k}{\epsilon}+\log m\cdot\log n)\right)$	 & $O((\epsilon^{-2} \log m + k) (m \log m + \log n))$ \\
			\textbf{Copeland}                & $O((\epsilon^{-1} \sqrt{k} m^2 + k) \log m n{\cdot\log m})$ & 	$O\left(\frac{k}{\epsilon}(m^{2}\log\frac{k}{\epsilon}+\log n)\right)$ & $O((\epsilon^{-2} \log m + k) (m \log m + \log n))$ \\
			\textbf{Cup}                & $O((\epsilon^{-1} \sqrt{k} m^2 + k) \log m n{\cdot\log m})$ & 		$O\left(\frac{k}{\epsilon}(m\log\frac{k}{\epsilon}+\log m\cdot\log n)\right)$ & $O((\epsilon^{-2} \log m + k) (m \log m + \log n))$ \\
			\textbf{Run Off}                 & $O((\epsilon^{-1} \sqrt{k} m^2 + k) \log m n{\cdot\log m})$
			& $O\left(\frac{k}{\epsilon}\log n\right)$
			& $O((\epsilon^{-2} + k) (m \log m + \log n))$ \\
			\textbf{Bucklin}                 & $O\big((\epsilon^{-1} \sqrt{k} m \log^2 m + k)$ & $O(\frac{k\log m}{\epsilon}(m\log\frac{k}{\epsilon} + \log n))$&
			$O((\epsilon^{-2} \log m + k) (m \log m + \log n))$ \\
			& \hspace{82pt}$\cdot \log mn{\cdot\log m}\big)$&& \\\hline 
			\multicolumn{4}{c}{\textbf{ Lower Bounds} \phantom{$2^{2^{2^2}}$}}      \\\hline
			\textbf{All rules} & \multicolumn{2}{l}{$\Omega((\epsilon^{-1} \sqrt{k} + k) \log n / \log k)$} \phantom{$2^{2^{2^2}}$}  & Already for $m=2$.     \\
			\textbf{Approval} & \multicolumn{2}{l}{	$\Omega\left(\epsilon^{-1}km\cdot\log\left(n/k\right)\right)$} & For deterministic protocols.
	\end{tabular}}
	\caption{Overview of our results. 
		$\epsilon$ is the required approximation, $k$ is the number of sites, $m$ is the number of candidates, and $n$ is the number of voters.
		There are three columns of upper bounds, where the first is for protocols based on counting frequencies,
		the second is for protocols based on checkpoints, and the third is for sampling-based protocols.
		The results in the first column and in the third column correspond to randomized protocols,
		while the results in the second column correspond to deterministic protocols. 
		For Plurality with run-off, the second protocol is actually a hybrid between checkpoints to (deterministic) frequency count.
		For Cup and Condorcet, one might also use the checkpoints protocol of Copeland.
		In our randomized lower bound we make the usual assumptions that there is no spontaneous communication.
	}
	\label{table:results}
\end{table*}
		
\begin{remark}
	Notice that we state our results for protocols which are correct with some constant probability, say $0.9$.
	One can always achieve arbitrary high probability $1-\delta$,
	as follows, and depending on the general technique used:
	\begin{itemize}
		
		\item
		For protocols based on counting frequencies,
		following the discussion in \Cref{section:protocols-based-on-counting-frequencies}, one can get failure probability $\delta$ by replacing the $\log m$ term with a $\log\frac m\delta$ term in the communication complexity.
		
		\item
		Protocols based on checkpoints are deterministic anyhow.
		
		\item
		For protocols based on sampling,
		we mention that,
		as can be seen from the corresponding proofs,
		the increase of the required sampling size needed for increasing the success probability is quite small.
		Specifically, the number of samples will increase:
		in \textsc{$t$-Approval} to $O(\eps^{-2}\log(\frac{2t}{\delta}))$, in \textsc{(Plurality with) Run Off} to $O(\eps^{-2}\log(\frac{1}{\delta}))$, and in all other voting rules to $O(\eps^{-2}\log(\frac{m}{\delta}))$.
		
	\end{itemize}
\end{remark}

\subsection{Approval-based Rules and Scoring Rules}

Let us begin with Plurality,
as it is arguably the simplest voting rule.
In Plurality,
a vote in a distributed vote stream is associated with one candidate out of the $m$ candidates participating in the election,
and the goal is for the center to maintain a candidate $c$ such that the highest number of voters vote for $c$,
or at least it is at most $\epsilon n$-far from being such a candidate.
Equivalently,
a distributed stream for Plurality contains $m$ item types (one item type for each candidate).
Given an approximate frequency for each type (that is, an approximate number of voters voting for each candidate),
the center can safely declare the candidate with the highest approximate frequency.

The next result follows by realizing a straight-forward protocol based on counting frequencies,
as described in \Cref{section:protocols-based-on-counting-frequencies}.

\begin{theorem}\label{theorem:plurality1}
  There is a protocol for \textsc{Plurality-winner-tracking}
  which uses $O((\epsilon^{-1} \sqrt{k} + k) \log n{\cdot\log m})$ words.
\end{theorem}

\begin{proof}
We use the efficient protocol for \textsc{Frequency-tracking}~\cite{huang2012randomized} with $\epsilon' = \epsilon / 2$.
This allows the center to maintain,
for each candidate $c$,
a value which is guaranteed to be at most $\frac{\epsilon}{2} n$-far from the real number of voters voting for $c$.
The center would declare the candidate $\hat{c}$ for which the approximate frequency is the highest.
A pseudocode of the protocol is given in Algorithms~\ref{algorithm:pluralityCenter}, \ref{algorithm:pluralitySite}.

\begin{algorithm}[t]
	\ \\
	Initiate and maintain approximate frequencies $f'(c)$ for each candidate using a \textsc{Frequency-tracking} protocol with $\epsilon' = \epsilon / 2$\; 
	\textbf{Upon query:} declare $\hat{c}$ for which $f'(\hat{c})$ is the highest\;
	\caption{Protocol for \textsc{Plurality-winner-tracking}, for center.}
	\label{algorithm:pluralityCenter}
\end{algorithm}

\begin{algorithm}[t]
	\ \\
	Participate in the \textsc{Frequency-tracking} protocol initiated by the center by reducing a vote for candidate $j$ to an item of type $j$\;
	\caption{Protocol for \textsc{Plurality-winner-tracking}, for site.}
	\label{algorithm:pluralitySite}
\end{algorithm}

Let us denote the real frequency of a candidate $c$ by $f(c)$ (which equals its Plurality score),
and its approximate frequency computed by the \textsc{Frequency-tracking} protocol by $f'(c)$.
For each $c \neq \hat{c}$,
it holds that 
\[
 f(c) \leq f'(c) + \frac{\epsilon}{2}n \leq f'(\hat{c}) + \frac{\epsilon}{2}n \leq f(\hat{c}) + \epsilon n
\]
where the first and third inequalities follows from the $\epsilon / 2$-approximation and the second from our choice of $\hat{c}$.
Therefore, we conclude that $\hat{c}$ is an $\epsilon$-winner, as required.
\end{proof}

We go on to consider $t$-Approval, where each voter specifies $t$ candidates which she approves.
We provide three protocols,
based on counting frequencies, checkpoints, and sampling, respectively.

\begin{remark}
Notice that here, as well as for other protocols which we discuss later on, we provide a different upper bound on the communication complexity of each of the protocols we describe. Note that these upper bounds are incomparable, in the sense that each of them is better for different specific values of the parameters. To not clutter the text too much, we defer a discussion on the interplay between the parameters, which relates to the decision of which protocol to choose for a specific scenario to \Cref{section:outlook}.
\end{remark}

\begin{theorem}\label{theorem:tapproval1}
  There are three protocols for \textsc{$t$-Approval-winner-tracking}, for $t \leq m / 2$.
  Respectively,
  the protocols use
  $O((\epsilon^{-1} \sqrt{k} t + k) \log tn{\cdot\log m})$,
  $O\left(\frac{k}{\epsilon}(m\log\frac{k}{\epsilon}+\log n)\right)$,
  and\\
  $O(\epsilon^{-2} \log (2t) + k) (\log{m\choose t} + \log n)$
  words of communication.
\end{theorem}
\begin{proof}

For the first protocol,
we reduce $t$-Approval to Plurality,
as follows, and as depicted in \Cref{figure:reduction-t-approval}.
Each site,
upon receiving a voter $v$ which approves $t$ candidates,
instead of considering the voter $v$,
creates and considers $t$ voters, $v_1, \ldots v_t$,
such that voter $v_i$ (for $i \in [t]$) is set to approve the $i$th approved candidate of $v$.
For example,
a voter approving $\{a, b, d\}$ would be reduced to three voters,
approving $a$, $b$, and $d$, respectively.
\begin{figure}
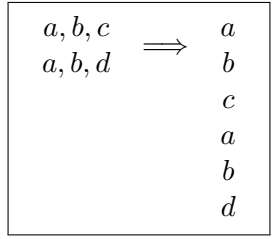

	\centering
	\fbox{
		\begin{tabular}{ c c c } 
			$a, b, c$ & \multirow{2}{*}{$\Longrightarrow$} & $a$ \\
			$a, b, d$ &                                    & $b$ \\
			&                                    & $c$ \\
			&                                    & $a$ \\
			&                                    & $b$ \\
			&                                    & $d$
		\end{tabular}
		
	}
	\caption{Reducing $t$-Approval to Plurality, for $t = 3$. Notice that two $t$-Approval voters become six Plurality voters.}
	\label{figure:reduction-t-approval}
\end{figure}

The reduced election has $n' = nt$ voters, and will be executed with precision parameter $\epsilon' = \epsilon / 2t$.
Consider a candidate $c$ which is an $\epsilon'$-winner in the reduced election;
we argue that $c$ is an $\epsilon$-winner in the original election.
Indeed, we can add $\epsilon n$ voters,
each approving $c$,
while for each other candidate $c'$,
at most $\epsilon n / 2$ of them approve $c'$ (as $t \leq m / 2$);
thus, the relative score of $c$ 
increases by $\epsilon n / 2=\epsilon' n' $. As $c$ is $\epsilon'$-winner in the reduced election, this is sufficient. By \Cref{theorem:plurality1}, the communication used is $O((\epsilon'^{-1} \sqrt{k} + k) \log n'{\cdot\log m})=O((\epsilon^{-1} \sqrt{k} t + k) \log tn{\cdot\log m})$.

The second protocol is based on checkpoints. We describe the static protocol for computing an $\eps$-winner.
The center initiates communication with all sites,
asking from each site to send an approximate score for each candidate.
That is, each site, for each candidate $c$,
sends the number of voters approving $c$,
rounded to the closest multiplication of $\epsilon n / k$.
Such rounding is enough,
since, summing up the possible errors from all $k$ sites,
the center would have a value which is at most $\epsilon n/2$-far from the real score.
Thus, the candidate $c$ with the highest approximated score will indeed be an $\epsilon$-winner. 
Each site should communicate $\log(\frac{k}{\epsilon})$ bits per candidate.
Thus,
the total communication is bounded by
$k\lceil\frac{m\log\frac{k}{\epsilon}}{\log n}\rceil\le O(k(1+\frac{m\log\frac{k}{\epsilon}}{\log n}))$.
The bound follows.

For the third protocol, we will show that $s=\frac{24}{\eps^2}\ln\frac{2t}{\delta}$ sampled voters,
chosen uniformly at random (with repetitions),
are enough to determine an $\eps$-winner with failure probability at most $\delta$. 
As we can communicate each voter using $\log {m \choose t}$ bits, the bound follows.
Consider such a sample of $s$ voters,
and,
for a candidate $c$, let $X_i^c$ be an indicator for the event that the $i$-th sampled voter approved $c$.
Let $X^c=\frac ns\sum_{i=1}^{s}X_i^c$, and denote by $Y^c$ the actual number of voters that approved $c$ in the original election.
Set $\mu=\mathbb{E}\left[\sum_{i}X_{i}^{c}\right]=s\cdot\frac{Y^{c}}{n}$. Using the Chernoff bound (\Cref{thm:chernoff} in \Cref{section:preliminaries}),
we have that:
\begin{align*}
\Pr\left[\left|X^{c}-Y^{c}\right|\ge\frac{\epsilon}{2}n\right] & =\Pr\left[\left|\sum_{i}X_{i}^{c}-s\cdot\frac{Y^{c}}{n}\right|\ge\frac{\epsilon}{2}s\right]\\
& =\Pr\left[\left|\sum_{i}X_{i}^{c}-\mu\right|\ge\frac{\epsilon}{2}\frac{n}{Y^{c}}\mu\right]\\
& \le2\exp\left(-\left(\frac{\epsilon}{2}\frac{n}{Y^{c}}\right)^{2}\cdot\mu/3\right)\\
& =2\exp\left(-\frac{\epsilon^{2}}{12}\cdot\frac{ns}{Y^{c}}\right)~.
\end{align*}
By union bound, 
we have that:
$$\Pr\left[\exists c\text{ s.t. }\left|X^{c}-Y^{c}\right|\ge\frac{\epsilon}{2}n\right]\le2\sum_{c}\exp\left(-\frac{\epsilon^{2}}{12}\cdot\frac{ns}{Y^{c}}\right)\le 2t\cdot e^{-\frac{\epsilon^{2}s}{12}}\le\delta~,$$ 
where the second inequality follows from \Cref{clm:tApprovalExp} below,
by setting $\lambda=\frac{\eps^2\cdot ns}{12}$ and noting that $(Y^{c_1},\dots,Y^{c_m})$ lies in the convex hull of the set $A$ described there.
The center will return a candidate $c$ with maximal $X^c$. 
Correctness follows by the same arguments as in the frequency-count protocol.
\end{proof}

\begin{claim}\label{clm:tApprovalExp}
  Consider the set $\mathbb{N}^m$ of points with $m$ integer coordinates.
  Let $A\subset \mathbb{N}^m$ contain exactly those points in $\mathbb{N}^m$
  for which the value of exactly $t$ coordinates is $n$, while the value of all their other $m - t$ coordinates is $0$.
  Let $\lambda \ge 2n$.
  Then,
  for any arbitrary point $(x_1,\dots,x_m)$ in the convex hull of $A$,
  it holds that:
  $$\sum_{i=1}^{m}e^{-\frac{\lambda}{x_{i}}}\le t\cdot e^{-\frac{\lambda}{n}}.$$ 
\end{claim}

\begin{proof}
Consider the function $f(x)=e^{-\frac\lambda x}$ and notice that its second derivative is $$\left(f(x)\right)''=\left(e^{-\frac{\lambda}{x}}\right)''=\left(\frac{\lambda}{x^{2}}\cdot e^{-\frac{\lambda}{x}}\right)'=-\frac{2\lambda}{x^{3}}\cdot e^{-\frac{\lambda}{x}}+\frac{\lambda^{2}}{x^{4}}\cdot e^{-\frac{\lambda}{x}}=\frac{\lambda}{x^{3}}\cdot e^{-\frac{\lambda}{x}}\cdot\left(\frac{\lambda}{x}-2\right)~.$$ 
Hence, $f$ is convex in the domain $[0,n]\subseteq[0,\lambda/2]$.
Set $\hat{f}(x_1,\dots,x_n)=\sum_{i=1}^n f(x_i)$. As sum of convex functions is also convex, $\hat{f}$ is convex in the domain $[1,n]^n$, which in particular  contains the convex hull of $A$.
Since $\hat{f}$ is convex function, the maximum value in the convex hull achieved in a point of $A$. We conclude that:
\[
\sum_{i=1}^{m}e^{-\frac{\lambda}{x_{i}}}=\hat{f}\left(x_{1},\dots,x_{m}\right)\le\max_{(y_{1},\dots,y_{m})\in A}\hat{f}(y_{1},\dots,y_{m})=t\cdot e^{-\frac{\lambda}{n}}~.\qedhere
\]
\end{proof}

For Approval,
where the set of approved candidates of each voter can be arbitrary,
thus upper bounded by the number $m$ of candidate,
we proceed similarly to $t$-Approval.
Naturally,
we have $m$-factors instead of $t$-factors in our bounds.
(Specifically, in the first protocol the size of the reduced election is $n' = mn$
and in the second protocol we sample slightly more voters.)

\begin{theorem}\label{theorem:approbalNOt}
  There are three protocols for \textsc{Approval-winner-tracking}.   
  Respectively,
  the protocols use
  $O((\epsilon^{-1} \sqrt{k} m + k) \log mn{\cdot\log m})$,
  $O\left(\frac{k}{\epsilon}(m\log\frac{k}{\epsilon}+\log n)\right)$,
  and
  $O((\epsilon^{-2} \log m + k) (m + \log n))$     
  words of communication.
\end{theorem}

We go on to consider ordinal elections.
Specifically,
next we consider the Borda rule,
for which we describe three protocols.

\begin{theorem}\label{theorem:borda1}
	There are three protocols for \textsc{Borda-winner-tracking}.  
	Respectively, the protocols use
	$O((\epsilon^{-1} \sqrt{k} m + k) \log mn{\cdot\log m})$,
	$O(\epsilon^{-1} k (m \log(k / \epsilon) + \log n))$, and 
	$O((\epsilon^{-2} \log m + k)$ $(m \log m + \log n))$ words of communication.
\end{theorem}

\begin{proof}
We start by discussing the impact of adding voters for elections whose results are decided by Borda.
For an arbitrary candidate $c$,
consider two voters where one voter is ranking $c$ first and then ranks the other candidates in an arbitrary order,
and another voter is ranking $c$ first and then ranks the other candidates in reverse order.
Adding these two voters causes an increase to the score of $c$ by $2(m - 1)$ while the score of all other candidates increases by $m - 2$.
Thus,
by adding $\epsilon n$ voters,
we can increase the relative score of $c$ by $\epsilon nm/2$.

The first protocol is based on reducing Borda to Plurality, similarly to the first protocol stated in \Cref{theorem:tapproval1}.
Specifically,
we begin by reducing Borda to Plurality,
as follows,
and as depicted in \Cref{figure:reduction-borda}:
  Each site,
upon receiving a voter $v$ with preference order $c_1 \pref \ldots \pref c_m$,
instead of considering the voter $v$,
creates and considers $\sum_{j \in [m]} m - j< m^2$ voters,
such that for $j \in [m]$, it creates $m - j$ voters,
each approving $c_j$.
For example,
a voter $v : a \pref b \pref d$ would be transformed into three voters,
approving $a$, $a$, $b$, respectively.

\begin{figure}
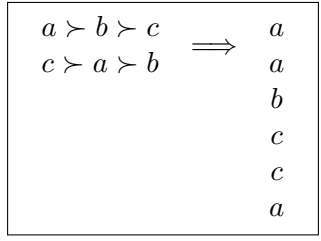

\centering
\fbox{
  \begin{tabular}{ c c c }
    $a \pref b \pref c$ & \multirow{2}{*}{$\Longrightarrow$} & $a$ \\
    $c \pref a \pref b$ &                                    & $a$ \\
                        &                                    & $b$ \\
                        &                                    & $c$ \\
                        &                                    & $c$ \\
                        &                                    & $a$
  \end{tabular}
}
\caption{Reducing Borda to Plurality. Notice that two Borda voters become six Plurality voters.}
\label{figure:reduction-borda}
\end{figure}

In the reduced election
we have $n' <m^2 n$ voters, where $n$ is the number of voters in the original election.
We use the protocol for Plurality described in \Cref{theorem:plurality1}
with $\epsilon' = \epsilon / (4m)$.
Let us denote the real frequency of a candidate $c$ in the reduced election by $f(c)$
and its computed approximate frequency by $f'(c)$.
The error is bounded by $\left|f'(c)-f(c)\right|\le {\epsilon'n'}<\frac{\epsilon}{4m}\cdot {nm^{2}}=\frac{\epsilon nm}{4}$.
Since by adding $\epsilon n$ voters we can increase the relative score of the chosen candidate by $\epsilon n m / 2$, we are done.

The second protocol is based on checkpoints,
thus below we describe the static subprotocol used in each checkpoint.
Similarly  to the second protocol in \Cref{theorem:tapproval1}, each site sends an approximation of the Borda score of each candidate rounded to the closest multiplication of $\epsilon nm / k$. Hence the subprotocol uses $O(k(1 + ({m\log\frac{k}{\epsilon}}) / ({\log n}) ))$ words, while the combined error for the Borda score estimation of each candidate is  $\epsilon nm / 2$.

For the third protocol, we will show that $s=O(\eps^{-2}\log\frac{m}{\delta})$ sampled voters, chosen uniformly at random (with repetitions), are enough to determine an $\eps$-winner with failure probability at most $\delta$. 
As we can communicate each voter using $\log (m!)$ bits, the bound follows.
For a candidate $c$, let $X_i^c=\frac{\alpha_i}{m}$, where $\alpha_i$ is the score that candidate $c$ gets from the $i$'s sampled voter.
Let $X^c=\frac{n\cdot m}s\sum_{i=1}^{s}X_i^c$, and denote by $Y^c$ the score of the candidate $c$ in the election.
Set $\mu=\mathbb{E}\left[\frac{1}{s}\sum_{i}X_{i}^{c}\right]=\frac{1}{n\cdot m}Y^{c}$. Using \Cref{lem:BasicSampler} we have that
\[
\Pr\left[\left|X^{c}-Y^{c}\right|\ge\frac{\eps}{4}\cdot n\cdot m\right]=\Pr\left[\left|\frac{1}{s}\sum_{i}X_{i}^{c}-\mu\right|\ge\frac{\eps}{4}\right]\le\frac{\delta}{m}~,
\]
and hence by union bound it follows that $\Pr\left[\exists c\text{ s.t. }\left|X^{c}-Y^{c}\right|\ge\frac{\eps}{4}\cdot n\cdot m\right]\le\delta$.
The center will return a candidate $c$ with maximal $X^c$.
The accuracy of the protocol follows from arguments given in the analysis of the frequency-count protocol.
\end{proof}

\begin{remark}
Results for other scoring rules, at least those corresponding to scoring vectors whose values are polynomially bounded,
can be achieved by similar techniques.
As the corresponding reductions to plurality are quite technical and do not provide new insights to the problem,
we do not consider them here.
\end{remark}

\subsection{Tournament-Based Rules}

In this section we consider Condorcet winners and the Copeland voting rule.
The rules we consider below are built upon the tournament defined over the election by considering head-to-head contests between all pairs of candidates.
The first protocol for Copeland proceeds by approximating,
for each pair of candidates $c_1$ and $c_2$, the number of voters preferring $c_1$ to $c_2$.
Having these approximate counts,
we will be able to identify an $\epsilon$-winner under Copeland.
If there is a candidate $c$ which is preferred to all other candidates,
then the center shall declare $c$ as the Condorcet winner.

\begin{theorem}\label{theorem:copeland1}
  There are three protocols for \textsc{Copeland-winner-tracking}.
  Respectively,
  the protocols use
  $O((\epsilon^{-1} \sqrt{k} m^2 + k) \log m n{\cdot\log m})$,
  $O\left(\frac{k}{\epsilon}(m^{2}\log\frac{k}{\epsilon}+\log n)\right)$,
  and
  $O((\epsilon^{-2} \log m + k) (m \log m + \log n))$ words.
\end{theorem}
  
\begin{proof}
For the first protocol,
we reduce each voter, corresponding to a total order over the candidates,
to $O(m^2)$ items;
specifically,
the reduced distributed stream will contain items of $O(m^2)$ item types,
where for each pair of candidates $c_1$ and $c_2$ we have a different type, denoted by $(c_1, c_2)$.
The reduction proceeds as follows.
Each site,
upon receiving a voter $v$ which specifies a linear order,
instead of considering the voter $v$,
creates and considers ${m \choose 2}$ items,
such that if $v$ prefers $c_1$ to $c_2$,
then we create an item $(c_1, c_2)$ (notice that this is an ordered tuple).
The reduction is depicted in \Cref{figure:reduction-copeland}.
For example,
a voter $v : a \pref b \pref d$ would be transformed into three items,
$(a, b)$, $(a, d)$, and $(b, d)$.

\begin{figure}
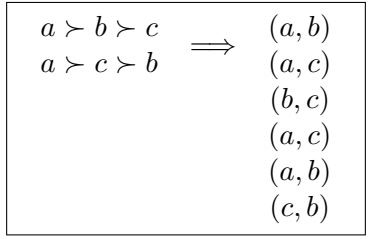

\centering
\fbox{
  \begin{tabular}{ c c c }
   $a \pref b \pref c$ & \multirow{2}{*}{$\Longrightarrow$} & $(a, b)$ \\
    $a \pref c \pref b$ &                                    & $(a, c)$ \\
                        &                                    & $(b, c)$ \\
                        &                                    & $(a, c)$ \\
                        &                                    & $(a, b)$ \\
                        &                                    & $(c, b)$
  \end{tabular}
}
\caption{Reducing a linear order to frequencies.}
\label{figure:reduction-copeland}
\end{figure}

The reduced distributed stream has $n' = {m\choose 2}\cdot n$ items
and $O(m^2)$ types of items.
For two candidates $c_1$ and $c_2$,
let $N(c_1, c_2)$ denote the number of voters preferring $c_1$ to $c_2$.
Now we can use a protocol based on counting frequencies (see \Cref{section:protocols-based-on-counting-frequencies}),
with $\epsilon' = \epsilon / m^2$,
to let the center maintain, for each pair of candidates $c_1$ and $c_2$,
a value $N'(c_1, c_2)$ such that $N'(c_1, c_2) \in N(c_1, c_2) \pm \epsilon' n' \subseteq N(c_1, c_2) \pm \epsilon n / 2$.

Let $\Sc'(c,E)$ be the number of candidates $c'$ such that $N'(c, c') \geq n / 2 - \epsilon n / 2$ in the election $E$.
We denote by $\Sc(c,E)$ the (real) Copeland score of candidate $c$ in election $E$. The center declares as an $\epsilon$-winner a candidate $\hat{c}$ with the highest value of $\Sc'(\hat{c},E)$.
Note that, for every candidate $c'$,
it holds that $\Sc(c',E)\le \Sc'(c',E)$;
this is so since the error in the computed frequency is bounded by $\eps n/2$, while for the declared winner $\hat{c}$,
it holds that there are at least $\Sc'(\hat{c},E)$ candidates $c'$ such that $N'(\hat{c}, c')\ge n/2- \epsilon n$.

Next we argue that $\hat{c}$ is indeed an $\epsilon$-winner.
We add $\epsilon n / 2$ voters which rank $\hat{c}$ on top and then the other candidates in arbitrary order,
and another $\epsilon n / 2$ voters which rank $\hat{c}$ on top and then the other candidates in reverse order.
Denote the modified election, with these additional voters, by $E'$.
Then,
for every $c'$, $N(\hat{c}, c')$ increased by $\eps n$; thus, $\Sc(\hat{c},E')\ge \Sc'(\hat{c},E)$.
Moreover, the number of wins of any other candidate $c'$ does not increase.
Hence  $\Sc(c',E')\le \Sc'(c',E)\le \Sc'(\hat{c},E)\le\Sc(\hat{c},E')$.

The communication complexity follows by the discussion given in \Cref{section:protocols-based-on-counting-frequencies};
specifically, it is $O((\epsilon'^{-1} \sqrt{k} + k) \log n'{\cdot\log m})=O((\frac{m^2}{\eps} \sqrt{k} + k) \log (nm){\cdot\log m})$.

The second protocol is based on checkpoints, and thus below we describe the static subprotocol used in each checkpoint. 
For every pair of candidates, $c_1$ and $c_2$,
every site sends, to the center, the number of voters preferring $c_1$ over $c_2$, rounded to the closest multiplication of $\eps n/2k$.
In each checkpoint, a candidate achieving estimated score higher that $\frac{n}{2}-\frac{\epsilon n}{2}$ for the maximal number of times 
(that is, for the largest number of other candidates) is declared a winner. As the error in each head-to head contest is upper-bounded by $k \cdot \frac{\eps n}{2k}=\frac{\eps n}{2}$, correctness follows by similar lines as given above in the proof of the frequency-count protocol. 
As there are $m^2$ quantities to estimate, 
each site sends $O\left(1+\frac{m^{2}\log\frac{k}{\epsilon}}{\log n}\right)$ words.
The total communication follows.

For the third protocol, we will show that $s=O(\eps^{-2}\log\frac{m}{\delta})$ sampled voters, chosen uniformly at random (with repetitions), are enough to determine an $\eps$-winner with failure probability at most $\delta$. 
As we can communicate each voter using $\log (m!)$ bits, the bound follows.
For two candidates $c,c'$, let $X_i^{(c,c')}$ be an indicator for the event that the $i$'s sampled voter prefers $c$ over $c'$. Let $N'(c,c')=\frac{n}{s}\sum_{i=1}^{s}X_i^{(c,c')}$, and denote by $N(c,c')$ the actual number of voters preferring $c$ over $c'$ in the original election.
Set $\mu=\mathbb{E}\left[\frac{1}{s}\sum_{i}X_{i}^{(c,c')}\right]=\frac{1}{n}N(c,c')$. Using \Cref{lem:BasicSampler} it follows that
\[
\Pr\left[\left|N'(c,c')-N(c,c')\right|\ge\frac\eps2\cdot n\right]=\Pr\left[\left|\frac{1}{s}\sum_{i}X_{i}^{(c,c')}-\mu\right|\ge\frac\eps2\right]\le\frac{\delta}{m^{2}}.
\]
By union bound,
with probability at least $1-\delta$,
for every pair of candidates we have that
$$\left|N'(c,c')-N(c,c')\right|<\frac\eps2\cdot n.$$ 
Let $\Sc'(c,E)$ be the number of candidates $c'$ such that $N'(c, c') \geq n / 2 - \epsilon n / 2$ in the election $E$.
The center declares as an $\epsilon$-winner a candidate $c$ with the highest value of $\Sc'(c,E)$. The accuracy of the protocol follows from arguments given in the analysis of the frequency-count protocol.
\end{proof}

We go on to consider the Cup rule, which differs from \textsc{Copeland} in several aspects. 
The first aspect is that, in order to prove that some estimated candidate $c$ is indeed an $\eps$-winner,
it is not enough to add $c$ arbitrary voters ranking $c$ last, but rather a more subtle construction of voters is needed.
The second aspect is that,
intuitively, while in Copeland we had to send communication regarding all pairs of candidates,
in Cup it is enough to send communication only regarding some pairs of candidates,
as given by the binary tree corresponding to the ``head-to-head'' contests performed for finding the winner under Cup.

\begin{theorem}\label{theorem:cup}
	There are three protocols for \textsc{Cup}.
	Respectively,
	the protocols use
	$O((\epsilon^{-1} \sqrt{k} m^2 + k) \log m n{\cdot\log m})$,
	$O\left(\frac{k}{\epsilon}(m\log\frac{k}{\epsilon}+\log m\cdot\log n)\right)$,
	and
	$O((\epsilon^{-2} \log m + k) (m \log m + \log n))$ words.
\end{theorem}

\begin{proof}
Let $T$ be an implementation of the binary tree of the \textsc{Cup} election:
  There are $n-1$ ordered pairs $P$ of candidates (corresponding to the head-to-head ``contests''),
  such that the winning candidate in each such pair goes up in the tree.
  In particular,
  every election $E$ which agrees with the tree $T$ on $P$, will have the root of $T$ as its \textsc{Cup}-winner.
We argue that there is an order $\pi_P$ over the candidates such that,
if $(c,c')\in P$,
then $c$ will appear before $c'$ in $\pi_P$.
Indeed, consider a directed graph $G$ with the candidates as its vertices and $P$ as its edges. $G$ is acyclic and thus a topological order of $G$ will provide us with the desired order $\pi$.
Later we will use this order $\pi$ as a preference order.
Now we will proceed to describing the protocols.

Our first protocol is based on counting frequencies, and is similar to the corresponding Copeland protocol.
We estimate the frequencies of \emph{all} head-to-head contests (using the same precision and communication). 
To return a winner, we simply run a \textsc{Cup} tournament (with the appropriate, given tree),
using the estimations $N'(c,c')$ instead of the real values $N_E(c,c')$.
As a result, we have a set $P$ of $n-1$ ordered pairs.
To prove correctness, it will be enough to show that by adding additional $\eps n$ votes it will hold,
for every $(c,c')\in P$,
that $N_{E'}(c,c')\ge N_{E'}(c',c)$.
Indeed, following the analysis of the frequency count of \textsc{Copeland},
with high probability for every pair of candidates $c,c'$ we have that $|N'(c,c')-N_E(c,c')|\le \eps n/2$.
Recall the order $\pi_P$ described at the beginning of the proof,
and notice that by adding $\eps n$ voters with preference orders as $\pi_P$ it will hold,
for every $(c,c')\in P$, that
\[
  N_{E'}(c,c')=N_{E}(c,c')+\epsilon n\ge N_{E}(c',c)=N_{E'}(c',c)~,
\]
as required.
	
The second protocol is based on checkpoints,\footnote{%
  The protocol described here is useful if we assume that $\log m\cdot\log n\ge m^{2}\log\frac{k}{\epsilon}$.
  If this is not the case, then we can use instead the communication protocol of \textsc{Copeland}.}
thus below we describe the static subprotocol carried-out in each checkpoint.
The subprotocol has $\log m$ rounds,
corresponding to the height of the binary tree associated with the Cup protocol.
In each round,
the center asks each site to provide approximate values of the pairs currently at interest.
Supplied with these approximate values,
the center then computes the winner of each head-to-head contest,
and continue to the nodes further up the tree.
At the end, the center declares the winner of the highest node in the tree.

More concretely,
for every pair of candidates of interest $c$, $c'$,
each site sends the center the number of voters preferring $c$ over $c'$, rounded to the closest multiplication of $\eps n/2k$.
As the error in each head-to-head contest is upper-bounded by $k \cdot \frac{\eps n}{2k}=\frac{\eps n}{2}$,
correctness follows by similar lines as given above in the proof of the frequency-count protocol described above. 
There are $\log m$ rounds, where at round $i$, each site sends $2^{\log m -i}$ values,
each requiring $\log \frac{2k}{\eps}$ bits.
Thus,
the total number of words in a checkpoint is: 
\[
k\cdot\sum_{i=1}^{\log m}\left\lceil \frac{2^{\log m-i}\cdot\log\frac{2k}{\epsilon}}{\log n}\right\rceil \le k\cdot\sum_{i=1}^{\log m}\left(1+\frac{2^{i}\cdot\log\frac{2k}{\epsilon}}{\log n}\right)=O\left(k\cdot\left(\log m+\frac{m\cdot\log\frac{k}{\epsilon}}{\log n}\right)\right),
\]
and total communication follows.

The third protocol is based on sampling and is similar to the Copeland sampling protocol.
We use the same communication, and hence we insure that with high probability, for every pair of candidates $c,c'$ it holds that
$|N'(c,c')-N(c,c')|<\frac\eps2\cdot n$.
Correctness now follows by similar lines as in the frequency-count protocol.
\end{proof}

Finally, we consider the Condorcet voting rule.
In order to declare a candidate $c$ as a Condorcet $\eps$-winner,
it is enough to insure that, by adding $\eps n$ voters, every other candidate $c'\ne c$ loses to at least one other candidate in the head-to-head contest (and thus, either $c$ can become a Condorcet winner in this way, or there will be no Condorcet winner at all, in which case $c$ can be returned).
A candidate $c$ which is either Copeland or Cup $\eps$-winner has this property.
We conclude that every protocol for Copeland as well as every protocol for Cup is in particular a protocol for Condorcet. 

\begin{corollary}\label{theorem:condorcet}
	There are three protocols for \textsc{Condorcet-winner-tracking}.
	Respectively, the protocols use
	$O((\epsilon^{-1} \sqrt{k} m^2 + k) \log m n{\cdot\log m})$, 
	$O\left(\frac{k}{\epsilon}(m\log\frac{k}{\epsilon}+\log m\cdot\log n)\right)$,
	and
	$O((\epsilon^{-2} \log m + k) (m \log m + \log n))$ words.
\end{corollary}

\subsection{Round-based Rules}

In this section we consider two round-based voting rules;
we begin with Plurality with run-off and then continue to Bucklin.
For Plurality with run-off we provide three protocols,
one of which is a ``hybrid'' protocol,
specifically combining checkpoints and sampling.
Intuitively,
hybrid protocols fit naturally with round-based voting rules,
which, informally speaking, are themselves ``hybrids'' of voting rules.

\begin{theorem}
  There are three protocols for \textsc{Plurality-with-run-off-winner-tracking}.
  Respectively, the protocols use  
  $O((\epsilon^{-1} \sqrt{k} m^2 + k) \log mn{\cdot\log m})$,
    $O\left(k\epsilon^{-1}\log n\right)$ and
    \linebreak
  $O((\epsilon^{-2} + k) (m \log m + \log n))$
 words.
\end{theorem}
\begin{proof}
The first protocol is based on counting frequencies.
We combine the protocol for Plurality, described in the proof of \Cref{theorem:plurality1},
with the protocol for Copeland, described in the proof of \Cref{theorem:copeland1}.
Specifically, the Plurality protocol maintains a frequency count for the plurality score of each candidate with accuracy $\frac{\eps}{6}$. The Condorcet protocol, for every two candidates $c_1,c_2$,  maintains a frequency count for the number of times $c_1$ wins $c_2$ with accuracy $\frac{\eps}{3}$. Following the analysis in \Cref{theorem:plurality1} and \Cref{theorem:copeland1}, the communication needed is  
$O((\epsilon^{-1} \sqrt{k}+ k) \log mn{\cdot\log m})+O((\epsilon^{-1} \sqrt{k} m^2 + k) \log mn{\cdot\log m})=O((\epsilon^{-1} \sqrt{k} m^2 + k) \log mn{\cdot\log m})$.
When calculating the winner,    
the center identifies two candidates $c_1,c_2$ with the highest estimated Plurality scores $f'(c)$ using the protocol for Plurality. Denoting by $f(c)$ the real Plurality score, for every $c'\ne c_1,c_2$ it holds that 
\begin{equation}\label{eq:runoff}
\mbox{For }i\in\{1,2\},~~f(c') \leq f'(c') + \frac{\epsilon}{6}n \leq f'(c_i) + \frac{\epsilon}{6}n \leq f(c_i) + \frac{\epsilon}{3}n~.
\end{equation}
Next, the center uses the protocol for Condorcet to decide which of these two candidates it shall declare as an $\epsilon$-winner.
Assume, without loss of generality, that it declares $c_1$ as the winner.
Then,
by adding $\frac{2}{3}\eps n$ (resp. $\frac{1}{3}\eps n$) voters ranking $c_1$ (resp. $c_2$) on top,
we can guarantee that $c_1$ and $c_2$ indeed have the highest Plurality score while $c_1$ wins $c_2$ in the head-to-head contest. 

The second protocol is a ``hybrid'' protocol which combines checkpoints and frequency count. 
During the protocol we maintain estimated frequencies of the Plurality score of each candidate as in the first protocol
(which we execute with precision $\frac{\eps}{6}$).
Next we describe the subprotocol executed in each checkpoint.
At each checkpoint, we use the Plurality protocol to identify two candidates $c_1$ and $c_2$ with the highest (approximated) Plurality score.
Given $c_1$ and $c_2$,
the center collects from all sites the \emph{exact} number of voters preferring $c_1$ over $c_2$,
and declares as a winner the one which is preferred by more voters.
Correctness follows as by adding $\frac{\eps}{2}n$ voters ranking $c_1$ on top, and $\frac{\eps}{2}n$ voters ranking $c_2$ on top,
we guarantee that $c_1$ and $c_2$ indeed have the highest plurality score (formally, this follows from equation (\ref{eq:runoff}))
while the winner between the two remains unchanged.
The subprotocol uses $2k$ words of communication, thus the total communication in all the checkpoints is $O\left(k\epsilon^{-1}\log n\right)$.
For the frequency count we will use the deterministic protocol with $O\left(k\epsilon^{-1}\log n\right)$ communication.
Therefore, in total we have a deterministic protocol with $O\left(k\epsilon^{-1}\log n\right)$ communication.

For the third protocol, we will show that $s=O(\eps^{-2}\log\frac{1}{\delta})$ sampled voters,
chosen uniformly at random (with repetitions),
are enough to determine an $\eps$-winner with failure probability at most $\delta$. 
We will use two sets of independent samples, $S_1$ and $S_2$, each of size $s/2=O(\eps^{-2}\log\frac{1}{\delta})$. 
According to the proof of \Cref{theorem:tapproval1} for the case of $t=1$,
the set of sampled voters $S_1$ is sufficient for us to determine the plurality score of each candidate with accuracy $\frac{\eps}{6}$.
Let $c_1$ and $c_2$ by the two candidates with the highest plurality score in $S_1$. 
Next we use $S_2$ to determine the number of times $c_1$ wins $c_2$
(with accuracy $\frac{\eps}{3}$ as in our protocol for Copeland; see \Cref{theorem:copeland1}),
and return the candidate who wins in the head-to-head contest (in $S_2$).
Correctness follows by similar lines to our frequency-count protocol.
\end{proof}

For Bucklin,
we suggest three protocols;
one is based on counting frequencies, the second is based on checkpoints, while the third is a sampling-based protocol.

\begin{theorem}
	 There are three protocols for \textsc{Bucklin-winner-tracking}.
	 Respectively,
	 the protocols use
	$O((\epsilon^{-1} \sqrt{k} m \log^2 m + k) \log mn{\cdot\log m})$,
	$O\left(\epsilon^{-1}\cdot k\cdot\log m\cdot\left(\log n+m\log\frac{k}{\epsilon}\right)\right)$,
    and \linebreak
	$O((\epsilon^{-2} \log m + k) (m \log m + \log n))$ words.
\end{theorem}

\begin{proof}
To make the proof ideas more clear, for simplicity we assume that $m$ is a power of $2$ (we mention that if this is not the case, than we can add less than $m$ dummy candidates, that always will be ranked after the real candidates).
By the pigeonhole principle, a Bucklin winner is necessarily found within the first $m / 2$ rounds.
This is so, since a candidate which is not a Bucklin winner appears less than $n / 2$ times within these $n \cdot m / 2$ positions, thus if no candidate is a Bucklin winner, then at least one position in the $n \cdot m / 2$ positions, constituting the $m / 2$ first positions of the $n$ voters, is not filled by any candidate.

We start with a discussion regarding the impact of adding voters. 
Let $c$ be an arbitrary candidate
and
consider adding two voters, each ranking $c$ on top,
and ranking the other candidates in reverse orders.
As a result,
the score of $c$ increases by $2$ for each level $j \le m / 2$,
while the status of each candidate $c' \ne c$ is only weaker
(thus, if $c'$ does not have a majority at level $j$ before the addition, then it will also not have a majority after the addition). 

The first protocol is based on counting frequencies.
It begins by reducing the distributed vote stream into a different distributed stream.
Intuitively,
the idea is to consider binary divisions of the positions between $1$ to $m$;
for example, if we know the frequency of some candidate $c$ in the first half positions (between position $1$ and position $m / 2$)
as well as the frequency of it in the positions between position $m / 2$ and position $3m / 4$,
then we know its frequency between position $1$ and $3m / 4$.
Thus,
we will have a different distributed stream for each binary division of the $m$ positions,
and we will use these to know the (approximated) Bucklin score of each candidate.

Formally, we will be estimating the frequencies of items of the form $(c,i,j)$, where $c$ is a candidate, $i\in[0,\log m]$, and $j\in[0,m/2^i-1]$. 
To this end,
each site,
upon receiving a voter $v$, increases (by $1$; this is equivalent to producing those entries and feeding them to the frequency count protocol) various entries in the frequency count protocol.
Specifically, for each candidate $c$ that was ranked at position $\ell \in [m]$, the site increases the items $(c, i, \left\lfloor \frac{\ell-1}{2^{i}}\right\rfloor)$, for all $i \in [0, \log m - 1]$  (alternatively  $(c, i,j)$ for $j \in [0, m / 2^i - 1]$,
such that $(j \cdot 2^i + 1)\le \ell \le ((j + 1) \cdot 2^i)$).
The idea is that we can recover the approximate number of voters ranking each candidate $c$ at the first $j$ positions
using $\log m$ approximate counters of these items.
See \Cref{figure:bucklin} for an illustrating example.

\begin{figure}
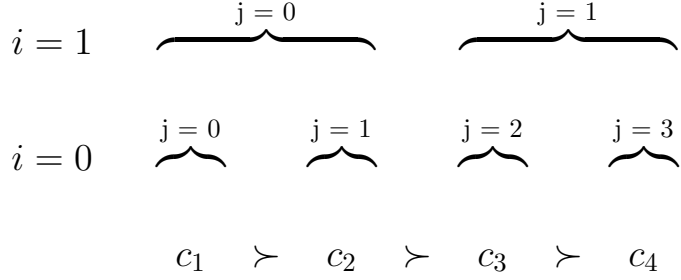
 
  \centering

{\Large
\[\begin{array}{*{20}{c}}

    i = 1 \ \ \ \ 
    & \multicolumn{3}{c}{\overbrace{\hphantom{.....................}}^\text{j = 0}} &
    & \multicolumn{3}{c}{\overbrace{\hphantom{.....................}}^\text{j = 1}}     \\ \ \\
    
    i = 0 \ \ \ \ 
    & \multicolumn{1}{c}{\overbrace{\quad \; }^\text{j = 0}}&
    & \multicolumn{1}{c}{\overbrace{\quad \; }^\text{j = 1}}&
    & \multicolumn{1}{c}{\overbrace{\quad \; }^\text{j = 2}}&
    & \multicolumn{1}{c}{\overbrace{\quad \; }^\text{j = 3}}&
    \\ \ \\
    
    & c_1 & \succ & c_2 & \succ & c_3 & \succ & c_4 \\
    
\end{array} \]}

  \caption{An example for the reduction performed in the protocol for Bucklin.
  Specifically, a voter $v : c_1 \pref c_2 \pref c_3 \pref c_4$ is considered,
  which is reduced to the following stream elements:
    $(c_1, 0, 0)$,
    $(c_1, 1, 0)$,
    $(c_2, 0, 1)$,
    $(c_2, 1, 0)$,
    $(c_3, 0, 2)$,
    $(c_3, 1, 1)$,
    $(c_4, 0, 3)$,
    $(c_4, 1, 1)$.
  }
  \label{figure:bucklin}
\end{figure}

The protocol initiates a \textsc{Frequency-tracking} protocol on the reduced distributed stream with $\epsilon' = \epsilon / (2m \log ^2 m)$.
(To arrive to the term above, notice that one $\log$ factor comes from the fact that each voter is replaced by $m \log m$ votes, while the second $\log$ factor comes to compensate for the fact that we estimate up to $\log m$ values.)
This will give us approximate values on the number of items of each type in our reduced distributed stream.
Let us denote,
for a candidate $c$ and position~$k \in [m]$,
the number of voters ranking $c$ at any position $k'\le k$ by $N(c, k)$.
Then,
we can approximate each of the values $N(c, k)$ by adding $\log m$ different approximated frequencies,
computed by the \textsc{Frequency-tracking} protocol (on the reduced stream).
This is, informally, the reason why we reduced each original voter in to the items we reduced to:
Given those items,
it is enough to add $\log m$ different approximated frequencies in order to approximate the value of $N(c, k)$;
then, as we will see below, bounding the error can be done in a finer way,
since the error is accumulated only in $\log m$ different frequencies,
and not in $m$ such (which would be the case otherwise).

Formally, for candidate $c$, index $i\in [m]$, and $j\in\{0,\dots,\frac{m}{2^i}-1\}$, denote by $f(c, i, j)$ the number of voters that ranked $c$ between the $(j \cdot 2^i + 1)$'th position and the $((j + 1) \cdot 2^i)$'th position. 
For an index $k\in[m]$, let $a_0,a_1,\dots,a_{\log m-1}\in\{0,1\}$ such that $k=\sum_{i=0}^{\log m-1}a_i\cdot 2^i$. It holds that $N(c, k)=\sum_{i=0}^{\log m-1}a_{i}\cdot f\left(c,i,\left\lfloor \frac{k-1}{2^{i}}\right\rfloor \right)$.
Denote by  $f'(c, i, j)$ the frequency count protocol estimation for $f(c, i, j)$.
Set $N'(c, k)=\sum_{i=0}^{\log m-1}a_{i}\cdot f'\left(c,i,\left\lfloor \frac{k-1}{2^{i}}\right\rfloor \right)$ to be the approximations of $N(c, k)$.

We are now able to simulate Bucklin;
specifically,
the center finds the minimum $k$ for which there is at least one candidate $c$ for which $N'(c, k) \geq \frac{n}{2}-\frac{\epsilon n}{2}$, and declares $c$ as an $\eps$-winner.

Next we show correctness.
The size of the reduced distributed stream is $n' = n m \log m$,
since each voter is transformed into $m \log m$ items, specifically $\log m$ per each candidate.
To approximate the value $N(c, k)$ we add up $\log m$ approximate frequencies,
each of which can be wrong by at most $\epsilon' n' = \epsilon n / (2\log m)$;
thus, the value of $N'(c, k)$ can be wrong by at most $\epsilon n/2$.
Therefore, in each level $j'<j$ where we do not find a winner,
there is indeed no candidate with a majority.
Finally, according to the discussion in the begging of the proof, $\epsilon n$ additional voters can indeed make our chosen candidate a winner.
A pseudocode of the protocol is given in Algorithms~\ref{algorithm:bucklinCenter}, \ref{algorithm:bucklinSite}.

\begin{algorithm}[t]
	\ \\
	Initiate and maintain a \textsc{Frequency-tracking} protocol with precision parameter  $\epsilon' = \epsilon / (2m \log ^2 m)$, where the items are all triples $(c,i,j)$ such that $c$ is a candidate, $i\in[0,\log m]$ and $j\in[0,m/2^i-1]$. 
	The estimated frequency count of  $(c,i,j)$ is denoted  $f'(c,i,j)$\;
	\ForEach{candidate $c$ and $k\in\{1,\dots,m\}$}{
		Let $a_0,a_1,\dots,a_{\log m-1}\in\{0,1\}$ such that $k=\sum_{i=0}^{\log m-1}a_i\cdot 2^i$\;
		Maintain  $N'(c,k)=\sum_{i=0}^{\log m-1}a_{i}\cdot f'\left(c,i,\left\lfloor \frac{k-1}{2^{i}}\right\rfloor \right)$\;
	}
	\textbf{Upon query:} Find the minimum $k$ for which there is at least one candidate $c$ for which $N'(c, k) \geq \frac{n}{2}-\frac{\epsilon n}{2}$, and return the corresponding $c$\;
	\caption{Frequencies based protocol for \textsc{Bucklin-winner-tracking}, for center.}
	\label{algorithm:bucklinCenter}
\end{algorithm}
\begin{algorithm}[t]
	\ \\
	Participate in the \textsc{Frequency-tracking} protocol initiated by the center\;
	Upon receiving a vote $v : c_1 \pref \dots \pref c_m$ reduce it to the following items:\\
	\For{$\ell\in\{1\dots,m\}$}{
		\For{$i\in\{0,\dots,\log m-1\}$}{
				Simulate insertion of item $\left(c_{\ell},i,\left\lfloor \frac{\ell-1}{2^{i}}\right\rfloor \right)$ to the \textsc{Frequency-tracking} protocol\;
			}
	
		}
	\caption{Frequencies based protocol for \textsc{Bucklin-winner-tracking}, for site.}
	\label{algorithm:bucklinSite}
\end{algorithm}

The second protocol is based on checkpoints, and thus below we describe the static subprotocol carried-out in each checkpoint.
Each checkpoint contains $\log m$ rounds,
where in each of these $\log m$ rounds,
the center is performing an approximate binary search to find the first $j$ for which there is at least one candidate $c_i$ 
for which the estimation of $N(c_i, j)$ is greater than $\frac{n}{2}-\frac{\epsilon n}{8}$,
and declares this $c_i$ as an $\epsilon$-winner.
In the round when some index $j$ is considered,
each site sends to the center the number of voters ranking each candidate $c$ among the first $j$ positions,
rounded to the closest multiplication of $\eps n/4k$.
Thus,
the center can estimate each $N(c_i, j)$ with precision $\frac{\epsilon n}{8}$, as needed.

Let $c$ be our declared candidate, which is declared at round $j$.
Then,
according to the discussion in the beginning of the proof,
by adding $\frac{\eps}{4}$ votes,
the declared candidate $c$ will have majority of the votes at round $j$, while no $c'\ne c$ will have majority of the votes for $j'<j$.
In particular the candidate $c$ is an $\frac\eps4$-winner.
Correctness follows by the discussion in \Cref{section:protocols-based-on-checkpoints}.
As at most $k\lceil\frac{m\log4k/\epsilon}{\log n}\rceil$ words of communication are required in each round of the sub-protocol,
and there are at most $\log m$ rounds, the total communication is bounded by
\[
O\left(\frac{\log n}{\epsilon}\cdot\log m\cdot k\lceil\frac{m\log4k/\epsilon}{\log n}\rceil\right)\le O\left(\frac{k\cdot\log m}{\epsilon}\cdot\left(\log n+m\log\frac{k}{\epsilon}\right)\right)~.
\]

For the third protocol, we will show that $s=O(\eps^{-2}\log\frac{m}{\delta})$ sampled voters,
chosen uniformly at random (with repetitions), are enough to determine an $\eps$-winner with failure probability at most $\delta$.
As we can communicate each voter using $\log (m!)$ bits,
the bound would follow.
So,
for each candidate $c$ and $j \in [m]$,
let $X_i^{(c,j)}$ be an indicator for the event that the $i$'s sampled voter ranks $c$ among the top $j$ positions.
Set $N'(c,j)=\frac ns\sum_{i=1}^{s}X_i^{(c,j)}$ to be an estimation for $N(c,j)$ - the number of voters ranking $c$ at among the top $j$ positions.
Using \Cref{lem:BasicSampler} we conclude that
$\Pr\left[\left|N'(c,j)-N(c,j)\right|\ge\frac\eps2\cdot n\right]\le\frac{\delta}{m^{2}}$.
By union bound, with probability at least $1-\delta$ for every all $c,j$, it holds that $\left|N'(c,j)-N(c,j)\right|<\frac\eps2\cdot n$. 
The center now finds the first $j$ for which there is at least one candidate $c$ for which $N'(c,j)\ge\frac{n}{2}-\frac{\epsilon n}{2} $, and declares this $c$ as an $\epsilon$-winner.
Correctness follows by the same arguments as in the frequency count protocol.
\end{proof}

\section{Lower Bounds}\label{section:lowerbounds}

In this section we provide lower bounds to complement the upper bounds derived above.
Our main result is an almost tight lower bound (up to a factor of $\log k{\cdot\log m}$)
for \textsc{Plurality-winner-tracking}.
We mention that our lower bound holds already for Plurality with $2$ candidates
and that it also improves the state-of-the-art lower bound for \textsc{Count-tracking}
(refer to \Cref{theorem:lowerbound} for our lower bound and to the remark which follows it for its application to \textsc{Count-tracking}).
Later in this section we describe a lower bound for deterministic protocols for \textsc{Approval-winner-tracking},
which is of some interest mainly since it is almost tight for \textsc{Approval-winner-tracking} and also shows that some dependency on the number $m$ of candidates
is required.

\subsection{Randomized Lower Bound for \textsc{Plurality-winner-tracking}}

Before we describe the randomized lower bound for \textsc{Plurality-winner-tracking},
we mention that it is applicable to all other voting rules we consider, via the following reduction.

\begin{lemma}
	Let $\calR$ be some voting rule described in \Cref{subsec:PrelimVotingRules}.
	A protocol for \textsc{$\calR$-winner-tracking} which uses $C$ words of communication
	implies a protocol for Plurality with $2$ candidates which uses $C$ words of communication.
\end{lemma}

\begin{proof}
Assuming a protocol for a voting rule $\calR$,
we can use it as a black-box for solving Plurality with $2$ candidates;
below we describe such a reduction.

Let $\calR$ be a voting rule considered in this paper.
Let $P$ be a protocol for $\calR$ which uses $C$ words of communication.
We construct a protocol $P'$ for Plurality with two candidates, $a$ and $b$,
which uses $P$ as a black-box.
Specifically,
we describe the operation of $P'$ for the different $\calR$'s considered in this paper;
the general idea of the reduction is similar for all these voting rules,
namely, given a Plurality election to construct a $\calR$ election where the $\calR$ winners
are equivalent to the Plurality winners.
The specifics of the reduction very slightly between the voting rules considered. 
We denote by $a$ and $b$ our two Plurality candidates.

If $\calR$ is Approval,
then,
for each Plurality voter which arrives and approves some candidate, say $a$,
we create a voter approving only $a$.
Notice that the the Approval winners are equivalent to the Plurality winners.

If $\calR$ is one of $\{$Borda, Condorcet, Copeland, Cup, Plurality with run-off, Bucklin$\}$,
then,
for each Plurality voter which arrives and approves $a$,
we create a voter ranking $a$ on top and then $b$;
similarly,
for each Plurality voter which arrives and approves $b$,
we create a voter ranking $b$ on top and then $a$.
Notice that, in these cases,
the $\calR$ winner are equivalent to the Plurality winners.

If $\calR$ is $t$-Approval,
then we shall create $2t - 2$ new candidates $c_1, \ldots, c_{2t - 2}$,
and,
for each Plurality voter which arrives and approves $a$,
we create one voter ranking $a, c_1, \ldots, c_{t - 1}$ on top,
and another voter ranking $a, b, c_t, \ldots, c_{2t - 3}$ on top;
similarly,
for each Plurality voter which arrives and approves $b$,
we create one voter ranking $b, c_1, \ldots, c_{t - 1}$ on top,
and another voter ranking $b, a, c_t, \ldots, c_{2t - 2}$ on top;
Notice that in this case also,
the $\calR$ winner are equivalent to the Plurality winners.
\end{proof}

We mention that in our lower bound for Plurality which we describe next,
we assume,
as it is usual in studying distributed streams,
that there is no spontaneous communication;
that is,
the center can initiate communication only as a result of receiving a message from the sites,
and each site can initiate communication only as a result of receiving a stream item or a message from the center.

Now we are ready to state our lower bound for Plurality;
the proof of the corresponding theorem (that is, \Cref{theorem:lowerbound})
appears at the end of the section,
and is based on \Cref{lem:lowerbound1} and \Cref{lem:lowerbound2}.
Recall that for \textsc{Plurality-winner-tracking}, \Cref{theorem:plurality1} provides an upper bound of $O((\epsilon^{-1} \sqrt{k} + k) \log n{\cdot\log m})$.

\begin{theorem}\label{theorem:lowerbound}
	Any randomized protocol for \textsc{Plurality-winner-tracking}
	uses at least $\Omega((\epsilon^{-1} \sqrt{k} + k) \log n / \log k)$ words of communication,
	even when there are only two candidates.
\end{theorem}

The next lemma shows a lower bound when $k < \epsilon^{-2}$.

\begin{lemma}\label{lem:lowerbound1}
  If $k < \epsilon^{-2}$,
  then any randomized protocol for \textsc{Plurality-winner-tracking}
	uses at least $\Omega(\epsilon^{-1} \sqrt{k} \log n)$ words of communication,
	even when there are only two candidates.
\end{lemma}

\begin{proof}
We reduce \textsc{Count-tracking} to \textsc{Plurality-winner-tracking}.
To this end,
we assume,
towards a contradiction,
that there is a protocol for \textsc{Plurality-winner-tracking} with $o(\epsilon^{-1} \sqrt{k} \log n)$
communication complexity,
and describe a protocol with the same communication complexity for \textsc{Count-tracking}.
For $k < \epsilon^{-2}$ this leads to a contradiction,
since there is a lower bound of $\Omega(\epsilon^{-1} \sqrt{k} \log n)$
for \textsc{Count-tracking} where $k < \epsilon^{-2}$~\cite[Theorem~2.4]{huang2012randomized}.
	
The distributed stream for \textsc{Count-tracking} contains items of only one type,
and a protocol for \textsc{Count-tracking} maintains a value $n'$
such that $n' \in n \pm \epsilon n$,
where $n$ is the number of items in the distributed stream.
We treat those items as voters, each of which is approving the candidate~$c_1$.

So, our reduction would work as follows:
  We assume the existence of a protocol for \textsc{Plurality-winner-tracking} in order to design a protocol for \textsc{Count-tracking}.
  Specifically, in order to design a protocol for \textsc{Count-tracking} with $k$ sites,
  we will use (the assumed) protocol for \textsc{Plurality-winner-tracking} with $k + 1$ sites.
The idea would be that the $k$ sites for the \textsc{Count-tracking} protocol would be mapped to $k$ sites of the \textsc{Plurality-winner-tracking} protocol,
where an item received by a site of the \textsc{Count-tracking} protocol would be mapped to a vote for candidate $c_1$, delivered to a site of the \textsc{Plurality-winner-tracking}.
Then, we center of the protocol for \textsc{Count-tracking} which we are designing would directly send votes for candidate $c_2$ to the remaining site, at certain times.
This would cause the protocol for \textsc{Plurality-winner-tracking} with $k + 1$ sites to swap winners back and forth from declaring $c_1$ to be a winner to declaring $c_2$ to be a winner;
intercepting those flips, we would be able to approximate the number of message, thus indeed describe a protocol for \textsc{Count-tracking}.

Below we explain how does the center operates in our protocol for \textsc{Count-tracking}.
The general idea of the reduction is as follows:
  For the center to simulate another site,
  called a \emph{ghost site} (we use this name to emphasize that it is not a ``real'' site, but just a ``virtual'' site which is being simulated by the center, as part of the center's internal computation),
  to which the center will send \emph{ghost voters} (again, not real voters, but only simulated by the center).
The center will simulate a protocol for Plurality with voters approving $c_1$ going to the $k$ ``real'' sites,
and simulated voters approving $c_2$ going to the ghost site.

Specifically,
the center has three parts (here, by parts we mean the implementation modules of the center; that is, the center is a center for \textsc{Count-tracking} which internally uses the following parts/modules in its internal computation).
The first part is a center for a \textsc{Plurality-winner-tracking} protocol operating on $k + 1$ sites.
The second part is a site in a \textsc{Plurality-winner-tracking} protocol; this is the ghost site.
The third part is for the center to inspect the \textsc{Plurality-winner-tracking} protocol from above,
and (using knowledge about the current winner) to send voters approving the candidate $c_2$ to the ghost site.

Let us denote the number of voters voting for $c_1$ ($c_2$) by $s(c_1)$ (respectively, $s(c_2)$).
Set $\delta = \epsilon / 10$.
The protocol for \textsc{Plurality-winner-tracking} will work with respect to approximation $\delta$,
and will consist of $k + 1$ sites.
Next we describe the logic of the third part of the center.
The estimation for \textsc{Count-tracking} will be $\est = (1 + 3 \delta) s(c_2)$
(note that only the ghost site receives voters approving $c_2$,
hence the center knows $s(c_2)$ exactly).
	
Before there is any communication from the (real) sites to the center,
we set $s(c_2) = 0$.
Then,
at some point in time there will be some communication from the sites to the center indicating that some voters approving $c_1$ arrived;
specifically, the first part of the center would declare $c_1$  as the winner of the election. 
More generally,
our protocol works in phases,
where a phase starts when the center ``flip''s its estimation;
that is, the (first part of the) center changes the estimation for the Plurality winner from $c_2$ to $c_1$.
When such a flip occurs,
the center sends some ghost voters (approving $c_2$) to the ghost site until
$s(c_2) = (1 + 3\delta)^i$ (for some $i$) and	a flip (from $c_1$ back to $c_2$) occurs.
(That is, we send ghost voters until a flip occurs and then send some additional voters until we reach a power of $1 + 3\delta$;
reaching this power of $1 + 3\delta$ is actually not needed, but it does not affect the communication complexity and it makes the analysis cleaner.)
We assume, as is usually done in distributed streams, that communication and internal computation happens instantly.
Thus, we have that $c_2$ is always the winner of the \textsc{Plurality-winner-tracking} protocol.
This finishes the description of the reduction.
	
Next we argue that our estimation (for \textsc{Count-tracking}) is accurate.
Specifically, we will show that $s(c_1) \le \est \le (1 + \epsilon) s(c_1)$.
As $c_2$ is always the winner,
it always holds that $s(c_{2}) + \delta(s(c_{1}) + s(c_{2})) \ge s(c_{1})$.
Since $\delta <  1 / 10$ (as $\epsilon < 1$),
it holds that:
$$s(c_{1})\le\frac{1+\delta}{1-\delta}\cdot s(c_{2}) < \left(1+3\delta\right)\cdot s(c_{2})=\est~.$$
Fix $s(c_2)=(1+3\delta)^i$.
Note that when $s(c_2)$ was equal to $(1+3\delta)^{i-1}$,
the protocol for \textsc{Plurality-winner-tracking} considered $c_1$ as the winner.
Hence,
$s(c_{1})+\delta(s(c_{1})+(1+3\delta)^{i-1})\ge(1+3\delta)^{i-1}$;
therefore,
\[
	s(c_{1})\ge\frac{1-\delta}{1+\delta}(1+3\delta)^{i-1}\ge\frac{(1+3\delta)}{(1+3\delta)^{3}}(1+3\delta)^{i}\ge\frac{\est}{1+\epsilon}~.
\]
Note that,
until the next flip,
$s(c_1)$ can only grow, while our estimation remains unchanged.
Hence, it will still hold that $\est\le (1+\epsilon) s(c_1)$.
Finally, we have that the communication of our protocol is bounded by
$o(\delta^{-1} \sqrt{k+1} \log (s(c_1)+s(c_2)))=o(\epsilon^{-1} \sqrt{k} \log n)$,
which contradicts the lower bound for \textsc{Count-tracking} discussed above.
\end{proof}

The next lemma is especially interesting for $k \ge \epsilon^{-2}$.

\begin{lemma}\label{lem:lowerbound2}
  Any randomized protocol for \textsc{Plurality-winner-tracking}
  uses at least $\Omega(k \log n / \log k)$ words of communication,
  even when there are only two candidates.
\end{lemma}

\begin{proof}
We assume that $\epsilon<\frac{1}{3}$.
Consider a protocol for \textsc{Plurality-winner-tracking} which is correct with probability $\frac23$ on every input.
Next we describe a distributed stream of voters which come to the sites.
Specifically,
the stream consists of $s$ phases.
Let $x_{1} = 1$,
$y_{1} = 1$,
$x_{i} = \left(1 + 3 \epsilon\right) \cdot k \cdot y_{i-1}$,
and $y_{i} = y_{i - 1} + x_{i}$.
During the $i$'th phase,
$x_{i}$ voters will go to each site and vote for the candidate $c_{\modd{i}}$. 
Note that after the $i$'th phase,
exactly $y_{i}$ voters voted at each site.
The total number of votes for $c_{\modd{i}}$ is at least $k\cdot x_{i}$,
while the total number of votes for $c_{\modd{i-1}}$ is at most $k\cdot y_{i - 1}$.
In particular,
$c_{\modd{i}}$ is a unique $\epsilon$-winner.

Note that
\[
y_{i}=y_{i-1}+x_{i}=y_{i-1}+\left(1+3\epsilon\right)\cdot k\cdot y_{i-1}=\left(1+\left(1+3\epsilon\right)\cdot k\right)\cdot y_{i-1}=\left(1+\left(1+3\epsilon\right)\cdot k\right)^{i-1}\cdot y_{1}\le\left(3k\right)^{i-1}~,
\]
thus the total number of voters is bounded by $n=k\cdot y_{s}<(3 k)^{s}$.
In particular, $s=\Omega(\frac{\log n}{\log k})$. 

Next consider the $j$'s site $S_j$ during the phase $i$.
Let $Y_{i,j}$ be the event that some communication between the center and $S_j$ occurs.
Let $Z_{i,j}$ be the event that the center initiates communication with $S_j$.
Let $X_{i,j}$ be the event that $S_j$ initiates communication with the center,
conditioned on the event that the center does not initiate communication with $S_j$
(that is, $Y_{i,j}$ conditioned on $\overline{Z_{i,j}}$).
We argue that $\mathbb{E}[X_{i,j}]=\Omega(1)$.
Before the $i$'th phases starts, $c_{\modd{i-1}}$ is the unique $\epsilon$-winner.

Consider an alternative scenario where, after the end of the $i-1$'th phase,
$x_{i}$ voters come to $S_j$
(and vote for $c_{\modd{i}}$),
while no additional voters arrive.
In this alternative scenario the center will not initiate communication with $S_j$,
as from its point of view nothing have changed since the end of the $(i-1)$'s phase
(since it did not receive any new messages).
Note also that in the alternative scenario,
$c_{\modd{i}}$ is the unique $\epsilon$-winner.
This is since
\begin{align*}
k\cdot y_{i-1}+\epsilon(k\cdot y_{i-1}+x_{i}) & =k\cdot y_{i-1}+\epsilon(k\cdot y_{i-1}+\left(1+3\epsilon\right)\cdot k\cdot y_{i-1})\\
& =k\cdot y_{i-1}\left(1+\epsilon(1+\left(1+3\epsilon\right))\right)\\
& =k\cdot y_{i-1}\left(1+2\epsilon+3\epsilon^{2}\right)<x_{i}.
\end{align*}

Thus,
if $S_j$ will not initiate communication with the center,
then,
in the alternative scenario,
the center would not hold the right estimation both at the end of the $i-1$'th phase and at the end of the $i$'th phase.
This is so since it will have the same estimation,
while there are different unique $\epsilon$-winners at those times.
Therefore,
the probability that the center is right in both of these times is bounded by $\Pr\left[X_{i,j}\right]$.
As the center has constant probability to have the right estimation twice,
we conclude that $\mathbb{E}[X_{i,j}]=\Omega(1)$.

Let us go back to our original scenario. 
Note that, from the point of view of $S_j$, both scenarios are the same (unless the center initiates communication). In particular $\Pr\left[X_{i,j}\mid\overline{Z_{i,j}}\right]=\Omega(1)$.
Set $\Pr\left[Z_{i,j}\right]=\alpha$.
Then, we have that:
\begin{align*}
\mathbb{E}\left[Y_{i,j}\right]=\Pr\left[Y_{i,j}\right] & =\Pr\left[Z_{i,j}\right]+\Pr\left[\overline{Z_{i,j}}\right]\Pr\left[X_{i,j}\mid\overline{Z_{i,j}}\right]\\
& =\alpha+(1-\alpha)\cdot\Omega(1)=\Omega(1)~.
\end{align*}
The total communication used during the whole protocol can be lower bounded by
\linebreak
$\sum_{i=1}^{s}\sum_{j=1}^{k}\mathbb{E}\left[Y_{i,j}\right]=\Omega(sk)=\Omega\left(\frac{k\log n}{\log k}\right)$.
\end{proof}

We are ready to prove \Cref{theorem:lowerbound}.

\begin{proof}[Proof of \Cref{theorem:lowerbound}~]
If $k<\epsilon^{-2}$, then \Cref{lem:lowerbound1}
provides us with a lower bound of $\Omega(\frac{\sqrt{k}}{\epsilon}\log n)=\Omega\left((\frac{\sqrt{k}}{\epsilon}+k)\frac{\log n}{\log k}\right)~.$
Otherwise ($k\ge \epsilon^{-2}$),
using \Cref{lem:lowerbound2} we get a lower bound of $\Omega(\frac{k\log n}{\log k})=\Omega\left((\frac{\sqrt{k}}{\epsilon}+k)\frac{\log n}{\log k}\right)$ .
\end{proof}

\begin{remark}\label{remark:byproduct}
Notice that \Cref{lem:lowerbound2} implies a $\Omega(\frac{k\log n}{\log k})$ lower bound for the \textsc{Count-tracking} problem.
The \textsc{Count-tracking} problem is a central problem in distributed streams, where the goal is to continuously maintain
a counter which is at most $\epsilon n$ far from the actual number of items arriving to the stream.
For the \textsc{Count-tracking} problem in the regime where $k\ge \epsilon^{-2}$,
Huang et al.~\cite[Theorem~2.3]{huang2012randomized} give a lower bound of $\Omega(k)$.

\Cref{lem:lowerbound2} relates to \textsc{Count-tracking},
As there is a reduction from \textsc{Plurality-winner-tracking} with two candidates to \textsc{Count-tracking}:
  to implement a protocol for \textsc{Plurality-winner-tracking} with two candidates
  it is sufficient to use two protocols for \textsc{Count-tracking} with $\epsilon' = \epsilon / 2$,
  one for each candidate,
  and to report as winner the candidate corresponding to the larger counter.

Thus,
we conclude that \Cref{lem:lowerbound2} implies a $\Omega(\frac{k\log n}{\log k})$ lower bound for the \textsc{Count-tracking} problem,
thus improving the state of the art for this problem.
\end{remark}

\subsection{Deterministic Lower Bound for \textsc{Approval-winner-tracking}}

Next we prove a lower bound on the communication of a deterministic protocol for \textsc{Approval-winner-tracking}. 
Recall the checkpoints-based deterministic protocol described within the proof of \Cref{theorem:approbalNOt}:
  the protocol has $O(\epsilon^{-1} \log n)$ checkpoints,
  and in each checkpoint,
  each site sends $\log(\frac{4k}{\epsilon})$ bits per candidate.
Thus,
if we measure the communication in bits (instead of words as in \Cref{theorem:approbalNOt}),
we get that the total cost of that protocol is $
O\left(\epsilon^{-1}\log n\cdot m\cdot k\cdot\log(\frac{4k}{\epsilon})\right)=
O\left(\frac{m\cdot k\cdot\log(\frac{k}{\epsilon})}{\epsilon}\cdot\log n\right)$.
In this section we prove \Cref{thm:LBapprobval},
showing that protocol (the one from \Cref{theorem:approbalNOt}) to be almost optimal in the deterministic regime. 

\begin{theorem}\label{thm:LBapprobval}
	For $\epsilon< 1/16$,
	and for large enough $m$,
	any deterministic protocol for \textsc{Approval-winner-tracking}
	uses at least
	$\Omega\left(\frac{mk}{\epsilon}\cdot\log\left(\frac{n}{k}\right)\right)$
	bits of communication.
\end{theorem}

The proof of \Cref{thm:LBapprobval} is based on a reduction from a new problem in \emph{communication complexity};
specifically, the variant of communication complexity which is sometimes referred to as \emph{multiparty communication complexity}.
In this variant we have $k$ players, denoted by $P_{1},\dots,P_{k}$,
and each player $P_j$ possesses a (possibly different) string $x_j\in \left\{ 0,1\right\} ^{m}$.
The objective is to compute the outcome of a function $f:\left\{ 0,1\right\} ^{m\times k}\rightarrow\{0,1\}$
on the combined inputs of the players
(formally, on the concatenation of the $x_j$ strings).
The players follow some protocol, and can communicate by broadcasting bits.
Specifically, when a player broadcasts a bit $b$,
all other players receive $b$ and we add $1$ to the communication count.
The cost of a protocol is the maximum number of exchanged bits, over all possible inputs.
The deterministic communication complexity of the function $f$,
denoted by $D(f)$,
is the minimal cost of a deterministic protocol that computes $f$.
For additional details and overview of the field we refer to
the textbook of Kushilevitz and Nisan~\cite{KN97}
or to the book chapter by Razborov~\cite{Razborov2011}.

Next we define the \emph{No Strict Majority} problem: \NSM, in short.
In it,
we have $2k$ players and a parameter $\eps>0$.
Each player $P_j$ has an $m$-bit string $A_j\in\{0,1\}^m$.
The objective is to figure out if there is an index $i$ such that a strict majority of the players has $1$ in that index.
Formally,
\[
\NSMm\left(A_{1},\dots,A_{2k}\right)=\begin{cases}
0 & \exists i\,\left|\left\{ j\mid i\in A_{j}\right\} \right|\ge\left(1+\epsilon\right)k\\
1 & \forall i\,\left|\left\{ j\mid i\in A_{j}\right\} \right|\le k\\
\mbox{Don't Care} & \mbox{Otherwise}
\end{cases}~~,
\]
where by ``Don't Care'',
we mean that any outcome of the protocol is legitimate.
The role of the ``Don't Care'' here is to allow us to reduce approximation problems to \NSM, as strict boundaries will not allow for that.

We denote a conjunction of $l$ instances of $\NSMm$ by $\bigwedge_{i=1}^{l}\NSMm$.
That is,
we have $2k$ players, each of which is given $l$ strings of $m$ bits each
(formally, $P_j$ gets $A_{j,1},\dots,A_{j,l}$);
the outcome shall be $1$ if and only if,
for every index $s\in[1,l]$ and $i\in [1,m]$,
it holds that
$\left|\left\{ j\mid i\in A_{j,s}\right\} \right|\le k$.
An equivalent way to think about $\bigwedge_{i=1}^{l}\NSMm$ is that each of the players gets a binary $l \times m$ matrix
and we accept if there is no cell for which a majority of the players has a $1$ in.

The proof of the following lemma could be found in \Cref{appendix:CC}.
We mention that,
as far as we know,
the $\bigwedge_{i=1}^{l}\NSMm$ problem was not considered in the literature,
hence the following lemma is novel and might be useful in other contexts besides our current context,
that of communication-efficient protocols for monitoring election winners.

\begin{lemma}\label{lem:LB-NSD}
	$D\left(\bigwedge_{i=1}^{l}\NSMf\right)=\Omega(mkl)$ .
\end{lemma}

To prove \Cref{thm:LBapprobval},
next we show how the communication complexity of $\bigwedge_{i=1}^{l}\NSMf$ implies a lower bound on the communication of \textsc{Approval-winner-tracking}.
The general idea,
similarly to the idea underlying the lower bound described in \Cref{lem:lowerbound2},
is to exploit the fact that,
in any point in time,
the center should be able to produce an answer without any additional communication.
Specifically,
we will have ${l=\Omega}(\frac{km}{\epsilon}\log \frac{n}{k})$ {rounds},
such that by sampling the center in $l$ different points of time we can determine $\bigwedge_{i=1}^{l}\NSMf$.

\begin{proof}[Proof of \Cref{thm:LBapprobval}]
For the sake of simplicity, during the proof we will consider also non-integer number of voters; this issue can easily be fixed by proper rounding, while introducing only a constant overhead to the number of voters.

Consider an instance of $\bigwedge_{i=1}^{l}\NSMf$,
where the input of player $P_j$ is $\left\{A_{j}^{s}\right\} _{s=1}^{l}\in\{0,1\}^{m\times l}$.
We will use a protocol for \textsc{Approval-winner-tracking} with $m$ candidates,
$2k$ sites,
and precision parameter $\eps$ to solve $\bigwedge_{i=1}^{l}\NSMf$.
By \Cref{lem:LB-NSD},
$\bigwedge_{i=1}^{l}\NSMf$ requires $\Omega(lmk)$ communication.
This in turn will imply a lower bound for the communication complexity of \textsc{Approval-winner-tracking}.

Our reduction is as follows.
Each player acts as a site,
and will simulate the arrival of voters in some order, to be specified shortly.
Player $P_1$ will act also as server (this is possible as we assume broadcast communication).
We denote the number of voters that approve candidate $i$ at site $j$ by $(v_i)_j$,
and the total number of voters, across all sites, approving candidate $i$ by $v_i=\sum_j(v_i)_j$.
The reduction has several phases. We first describe the first phase and later generalize it to describe how
the $r$th phase is executed.

\begin{itemize}

\item
\textbf{First phase:}
Before the first phase starts,
the situation is that each candidate is approved by $0$ voters at each site.
The first phase have 3 stages, as follows.

  \begin{itemize}
    \item
    Vote simulation: 
      each site $j$ simulates that a voter approving $A^1_j$ arrives.
    \item
    Validation:
      the center computes a winner $q$,
      then it collects $(v_q)_j$ from all the sites (players).
      If $v_q=\sum_{j=1}^{2k}(v_q)_j>k$,
      then it determines that the solution for the first instance is $0$.
      Otherwise it determines that the solution is $1$.
		\item
		Reset:
		  each site $j$ simulates that a voter approving $\overline{A^1_j}$ comes.
  \end{itemize}
  
\item
\textbf{$r$th phase:}
Set $x_r=(1+32\eps)^{r-2}$ and $y_r=32\eps x_r>\frac{16\eps}{1-8\eps}x_r$.
Before the $r$th phase starts,
the situation is that each site already received exactly $2\cdot x_r$ voters,
such that each candidate was approved by exactly $x_r$ voters at each site.
The $r$th phase has 3 stages:

	\begin{itemize}
		\item
		Vote simulation:
		  each site $j$ simulates that $y_r$ voters appeared,
		  all approving $A^r_j$.
		\item
		Validation:
		  the center computes a winner $q$,
		  then it collects $(v_q)_j$ from all the sites (players).
		  If $v_q=\sum_{j=1}^{2k}(v_q)_j>2k\cdot x_r+k\cdot y_r $,
		  then it determines that the solution for the $r$th phase is $0$.
		  Otherwise it determines that the solution is $1$.
		\item
		Reset:
		  each site $j$ simulates that $y_r$ voters appeared,
		  all approving $\overline{A^r_j}$.
	\end{itemize}
\end{itemize}

In total,
the number of voters used throughout the protocol is $n=2k\cdot2\cdot x_{l+1}=4k\cdot(1+32\eps)^{l{-1}}$.
In addition to the protocol for \textsc{Approval-winner-tracking},
we also used $O(k)$ communication in each phase to compute the number of votes the winner got;
to see why $O(k)$ bits of communication suffices for each phase,
notice that in phase $r$,
in the validation stage,
each site sends to the center the number of voters voted for the $q$'th candidate.
As there are only two options for this number ($x_r,x_r+y_r$), one bit of communication suffices for each site,
thus we have $O(k)$ additional bits of communication in total for each phase. 
Thus,
the total communication our protocol uses,
in addition to the \textsc{Approval-winner-tracking} protocol,
is $O(lk)$.

Next we argue that we indeed compute the right answer for each of the $l$ instances of $\bigwedge_{i=1}^{l}\NSMf$.
Note that,
at the time of the second step in the $r$th phase,
exactly $2x_r+y_r$ voters arrived at each site,
accumulating to a total of $n_r=2k\cdot (2x_r+y_r)$ voters.
Fix some $r\in[1,l]$ and consider first the case where there exists an index $i$ such that
$\left|\left\{ j\mid i\in A_{r}^j\right\} \right|\ge\left(1+\frac{1}{4}\right)k$.
In particular,
the $i$'th candidate was approved by at least
$v_{i}\ge2k\cdot x_{r}+(1+\frac{1}{4})\cdot k\cdot y_{r}$
voters.
Hence,
the Approval protocol will return an index $q$ s.t. $v_{q}+\epsilon\cdot n_{r}\ge v_{i}$.
As {$y_r>\frac{16\eps}{1-8\eps}x_r$ it holds that} $\frac{1}{4}\cdot k\cdot y_{r}>\epsilon\cdot n_{r}$,
{and in particular} $v_{q}\ge v_{i}-\epsilon\cdot n_{r}>2k\cdot x_{r}+k\cdot y_{r}$.
We conclude that,
in this case,
the algorithm will compute the correct answer in the $r$th phase.
	
Otherwise,
if for every index $i$,
we have that $\left|\left\{ j\mid i\in A_{j}^r\right\} \right|\le k$,
then no matter which index $q$ the algorithm for \textsc{Approval-winner-tracking} will return,
since the center will check it and will find out that $v_{q}\le2k\cdot x_{r}+k\cdot y_{r}$.
Hence, again, it will compute the right answer.
	
Note that the number of voters used throughout the protocol is $n=4k\cdot(1+32\eps)^{l{-1}}$,
hence $l={1+}\log_{1+32\eps}\frac{n}{4k}={\Omega}\left(\frac{1}{\epsilon}\log\frac{n}{k}\right)$.
As,
other then the protocol for \textsc{Approval-winner-tracking},
we used only $O(lk)$ bits,
while we solved $\bigwedge_{i=1}^{l}\NSMf$, a problem requiring $\Omega(mkl)$ bits, we conclude that \textsc{Approval-winner-tracking} requires at least 
	\[
		\Omega(lmk)-O(lk)=\Omega(lmk)=\Omega\left(\frac{mk}{\epsilon}\cdot\log\left(\frac{n}{k}\right)\right)~,
	\]
	bits.
In the first equality we used the fact that $m$ is large enough.
\end{proof}

\section{Discussion and Outlook}\label{section:outlook}

In this paper we studied communication-efficient protocols for maintaining approximate winners in distributed vote streams.
We have shown several general techniques for designing such protocols
(namely, sampling-based protocols, protocols based on checkpoints, and protocols based on counting frequencies),
and demonstrated their usefulness for various single winner voting rules.
Indeed,
based on these general techniques, for each of the rules we considered here,
we have designed several communication-efficient protocols,
and analyzed their communication complexity.
We complemented our protocols with lower bounds.

As a further contribution,
we view our paper as a bridge between issues and ideas from artificial intelligence
(specifically, multiagent systems and computational social choice)
and techniques and methods from theoretical computer science and database systems
(specifically, streaming and sampling algorithms and distributed continuous monitoring).
We hope that more fruitful research can be done by bridging between those fields.

Below we first discuss several aspects which are somehow hidden in the technical part of the part.
Specifically,
we begin with a discussion on deterministic protocols,
showing that,
while the technical part of the paper concentrates on randomized protocols,
communication-efficient deterministic protocols for monitoring election winners in distributed streams exist as well.
Then,
as in this paper we developed several protocols for each voting rule considered,
we provide a discussion on how to choose which protocol to use at which scenario,
depending on the specific parameters of the problem at hand.
We end this section by mentioning some directions for future research.

\subsection{Deterministic Protocols}

While in this paper we concentrated on randomized protocols,
it turns out that some of our protocols are already deterministic or can be made deterministic with some slight modifications.
To us, this is quite surprising:
  for example, there are usually no efficient deterministic algorithms operating on centralized streams.
Specifically, as we show next, while there are no natural deterministic equivalents to our sampling-based protocols
(since, informally speaking, a deterministic equivalent to sampling would basically need to sample the whole electorate),
our other protocols can generally be made deterministic.

Indeed, protocols based on checkpoints are already deterministic.
Further,
protocols based on counting frequencies can use a deterministic protocol for \textsc{Frequency count}
which uses $O(\epsilon^{-1} k \log n)$ words of communication~\cite{yi2013optimal}.
Correspondingly, the increase in the communication complexity is by at most a factor of $\sqrt{k}$.
Notice that the corresponding deterministic protocols still maintain only approximate solutions.

\subsection{Choice of Protocol}

A closer look at our upper bounds reveals that
the choice of which protocol to use for which voting rule
crucially depends on the relationships between the various parameters;
specifically,
as a rule of thumb,
it looks as if the choice of which protocol to use depends on the relation between $k$ and $1/\epsilon^{-2}$;
specifically, if $k < 1/\epsilon^{-2}$, then protocols based on counting frequencies or on checkpoints shall be used,
while if $k \geq 1/\epsilon^{-2}$, then sampling-based protocols achieve better communication complexity.
We believe that both cases make sense;
for example, in a supermarket chain with $4000$ stores,
requiring approximation of $\epsilon = 1/100$ would put us in the first case,
while requiring $\epsilon = 1/10$ would put us in the second case.

\subsection{Future Directions}

Below,
we discuss several directions for future research.

\subsubsection{Improved Bounds and More Rules}

While we considered quite a variety of voting rules in this paper,
there are further interesting rules to consider,
ranging from single-winner voting rules such as Kemeny, Young, Dodgson, Schulze, Maximin, and Ranked pairs,
to multiwinner voting rules such as committee scoring rules, including Chamberlin--Courant and Monroe.
Further,
there are still some gaps between our upper bounds and lower bounds;
closing those gaps is a natural direction for future research.

\subsubsection{Simulations and Heuristics}

Our focus in the current paper,
besides bridging between the study of computational social choice within the field of artificial intelligence
and the topic of continuous distributed monitoring within database systems and theoretical computer science,
is a theoretic study of communication-efficient protocols for maintaining election winners in distributed elections.

We believe that a theoretic study is important but also appreciate the possibility of validating our theoretical findings
by performing simulations. Thus we view an experimental follow-up to the current paper as an important and interesting future work.
One shall be careful in choosing input instances and evaluation methods,
and there is also some hope that efficient heuristics (for which the theoretical complexity might not be impressive)
outperform our protocols for certain scenarios and distributions.

\subsubsection{Constrained Resources}

In this paper,
we measured the complexity of our protocols only in terms of their communication cost.
It is natural to consider other resources,
especially studying various trade-offs between space, time, and communication.
We mention that,
for example,
our sampling-based protocols do extend to situations where the computational power of the sites is very limited,
since sampling from a distributed stream can be done with sites which have only logarithmic space~\cite{cormode2012continuous}.
Our checkpoint-based protocols, however, generally assume linear space (in $m$) for each site.

\subsubsection{Various Restrictions}

In this paper we have concentrated on worst-case notions:
first,
we assumed that voters are arbitrarily (thus, adversarially) assigned into the sites;
second,
we did not assume any structure on the electorate itself.
Since there might be better real-world situations,
it is natural to study protocols for elections drawn from, say, Mallow's model or the Urn model,
as well as to study situations where the voters are, say, uniformly assigned into the sites.
Of course, studying protocols for elections which adhere to some domain restrictions,
such as single peaked elections and single crossing elections would be natural and interesting as well.
Indeed, there is hope more efficient protocols exist for such restrictions.

\subsubsection{Incomplete Votes}

In this paper we considered situations in which voters provide full votes;
e.g., in the ordinal elections we consider, each voter is assumed to provide a full ranking.
There are quite a few papers dealing with settings in which voters provide only partial rankings:
  Pini et al.~\cite{pini2011incompleteness} and Xia and Conitzer~\cite{xia2011determining} study the complexity of computing \emph{possible} winners (alternatives for which at least one completion of the votes makes them win the election) and \emph{necessary} winners (alternatives for which all completions of the votes make them win the elections);
  Bentert and Skowron~\cite{bentert2019comparing} and Caragiannis et al.~\cite{caragiannis2019optimizing} identify voting rules which are suitable to situations in which the vote elicitation is incomplete, in the sense that they approximate well the winners for the full elicitation.
It would be interesting to try to generalize our communication protocols to such settings in which voters provide only partial ranking.

\section*{Acknowledgements}

The authors thank Robert Krauthgamer for inspiring discussions.

\bibliographystyle{alpha}
\bibliography{bib}

\appendix

\section{Proof of \Cref{lem:Checkpoints}}\label{appendix:proofOfCheckPoints}

\begin{proof}[Proof of \Cref{lem:Checkpoints}]
Set $\delta=\frac{\eps}{4}$. 
As $c$ is a $\delta$-winner in $E$,
it follows that there exist a set of voters $u_{1},\dots,u_{q'}$,
where $q'\le \delta n$,
such that $c\in\calR(\tilde{E})$ for $\tilde{E}=E\cup \{u_{1},\dots,u_{q'}\}$:
  that is, adding those $q'$ voters to $E$ would make $c$ a winner.
The situation is that we have an additional $q$ voters, $v_{n+1}, \ldots, v_{n+q}$,
which might have a bad impact with respect to $c$.
Thus,
our goal is to describe an additional set of $q'' \leq 3q$ voters,
denoted by $W = \{w_1, \ldots, w_{q'' \leq 3q}\}$ which will nullify the (possibly) bad impact of those $q$ voters
(which arrived after the last checkpoint)
on $c$.

So,
for each of the voting rules we consider in this paper,
we will argue that 
$c\in\calR(\tilde{E}')$
where
$\tilde{E}'=
	E'\cup \{w_1,\dots,w_{q''}\}\cup \{u_1,\dots,u_{q'}\}=
	\tilde{E}\cup\{v_{n+1},\dots,v_{n+q}\}\cup  \{w_1,\dots,w_{q''}\}$.
Thus, we will conclude that $c$ is a $4\delta =\eps$-winner with respect to $E'$.
Below we describe the set of voters $W$ for each voting rule separately.

\begin{itemize}

\item
\textbf{Plurality, $t$-Approval, Approval}:
For $i\in[q]$,
let $w_i$ be a voter approving $c$,
and such that $w_i$ is not approving any candidate which was approved by $v_{n+i}$
(recall that in the case of $t$-Approval we assume $t\le m/2$). 

As $c$ is a winner in both the elections with voters $\{v_{n+1},\dots,v_{n+q},w_{1},\dots,w_q\}$ and $\tilde{E}$,
it holds that $c\in\calR(\tilde{E}')$.
	
\item
\textbf{Borda}:
For $i\in[q]$,
let $w_i$ be the ``reverse'' of $v_{n+i}$
(e.g., if $v_{n+i} : a \pref b \pref c$, then $w_i : c \pref b \pref a$).
Note that all candidates have the same Borda score with respect to the voters
$\{v_{n+1},\dots,v_{n+q},w_{1},\dots,w_q\}$.
Thus $c\in\calR(\tilde{E})$ implies $c\in\calR(\tilde{E}')$. 
		
\item
\textbf{Cup}:
Denote the set of candidates by $M$ and consider the election $\tilde{E}$.
In order to compute a Cup-winner,
we shall preform a series of $m-1$ ``head-to-head'' contests.
That is,
there is a set $P\subseteq  M\times M$ of ordered pairs,
of size $m - 1$,
such that for every $(c_1,c_2)\in P$,
$c_1$ wins $c_2$ in an head-to-head contest.
In fact,
in any election $\hat{E}$ such that for every $(c_1,c_2)\in P$,
$c_1$ wins $c_2$ in an head-to-head contest with respect to $\hat{E}$,
it holds that $c$ is a Cup-winner.

In the beginning of the proof of \Cref{theorem:cup} we argued that there is an order $\pi_P$ over $M$,
such that for every $(c_1,c_2)\in P$, $c_1$ precedes $c_2$ in $\pi_P$.
Next we define $w_1,\dots,w_{q}$.
All these voters will order the candidates with respect to $\pi_P$:
  that is,
  the maximal candidate in $\pi_P$ will be ranked first, the second will be ranked second, and so on.
Now,
for every $(c_1,c_2)\in P$, $c_1$ wins $c_2$ in the ``head-to-head'' contest with respect to
$\{v_{n+1},\dots,v_{n+q},w_{1},\dots,w_q\}$,
hence $c_1$ wins $c_2$ in the ``head-to-head'' contest with respect to $\tilde{E}'$.
We conclude that $c\in\calR(\tilde{E}')$.
		
\item
\textbf{Copeland and Condorcet}: 
We prove the claim for Copeland first.
For $i\in[q]$, let $u_i$ be the ``reverse'' of $v_i$.
For every two candidates $c_1,c_2$, a majority of the voters prefer $c_1$ to $c_2$ with respect to $\tilde{E}$
if and only if
a majority of the voters prefer $c_1$ to $c_2$ with respect to $\tilde{E}'$.
Thus,
$c\in\calR(\tilde{E})$ implies $c\in\calR(\tilde{E}')$. 
The above proves the claim for Copeland;
thus the claim for Condorcet follows,
as every Copeland winner is in particular a Condorcet winner.
		
\item \textbf{Bucklin}:
For $i\in[q]$,
let $w_i$ be a voter ranking $c$ on top,
and such that every candidate $c'\ne c$,
which is ranked at position $j$ in $v_{n+i}$,
will be ranked at position $m-j+1$ or $m-j+2$ in $v_j$.
Note that,
for every $j\le\frac{m}{2}$ and $c'\ne c$,
the number of voters among $v_{n+1},\dots,v_{n+q},w_1,\dots w_q$ ranking $c'$ among the first $j$ positions is at most $q$,
while the number of voters ranking $c$ among the first $j$ position is at least $q$.

Set $n'=n+q'+2q$.
Suppose that in the elections $\tilde{E}$, $c$ wins at round $j$. 
Consider the election $\tilde{E}'$.
Then,
for every candidate $c'\ne c$ and $j'<j$,
the number of voters ranking $c'$ among the first $j'$ positions is less than $\frac{n+q'}{2}+q=\frac{n'}{2}$,
while the number of voters ranking $c$ among the first $j$ positions is at least $\frac{n+q'}{2}+q=\frac{n'}{2}$.
We conclude that $c\in\calR(\tilde{E}')$.		
		
\item \textbf{Run Off}:
Let $c'$ be a candidate such that $c$ and $c'$ get the highest plurality score in $\tilde{E}$,
and such that $c$ is winning $c'$ in the ``head-to-head'' contest with respect to $\tilde{E}$. 
Set $w_1,\dots,w_{q}$ to be voters ranking $c'$ on top, and set $w_{q+1},\dots,w_{3q}$ to be voters ranking $c$ on top. 
Note that with respect to the voters $v_{n+1},\dots,v_{n+q},w_1,\dots,w_{3q}$, $c$ and $c'$ have the highest plurality score,
while $c$ is winning over $c'$ in the ``head-to-head'' contest.
Thus this is also the situation in $\tilde{E}'$.
We conclude that  $c\in\calR(\tilde{E}')$.~\qedhere

\end{itemize}	
\end{proof}

\section{Communication Lower Bound for \textsc{No Strict Majority}}\label{appendix:CC}

A basic machinery for communication complexity lower bounds is fooling sets.
Consider a function $f:\{0,1\}^{m\times k}\rightarrow \{0,1\}$.
We have $k$ players, each holding a string from $\{0,1\}^m$.

\begin{definition}[Fooling set]
	A set $A=\{(x^1_1,\dots,x^1_k),\dots,(x^s_1,\dots,x^s_k)\}\subseteq  \{0,1\}^{m\times k}$ is called a \emph{fooling set}
	for the function $f:\{0,1\}^{m\times k}\rightarrow \{0,1\}$, if there are some bit $b\in \{0,1\}$ such that:
	\begin{enumerate}
		\item For every $i$, $f(x^i_1,\dots,x^i_k)=b$.
		\item For every $i\ne j$, there is $(y_1,\dots,y_k)\in \{x^i_1,x^j_1\}\times\cdots\times\{x^i_k,x^j_k\}$
      		such that  $f(y_1,\dots,y_k)\ne b$.
	\end{enumerate}  
\end{definition}

A fooling set is called a \emph{$1$-fooling set} if the bit $b$ above is $1$ (similarly, a \emph{$0$-fooling set}).
The proof of the following fact can be found in the textbook by Kushilevitz and Nisan~\cite{KN97}.

\begin{fact}\label{fact:HighFoolingSet}
	Let $f:\{0,1\}^{n\times k}\rightarrow \{0,1\}$ be some function with fooling set $A$.
	Then, $D(f)\ge\log|A|$.
\end{fact}

We denote by $f^l=\bigwedge_{i=1}^{l}f:\{0,1\}^{n\times k\times l}\rightarrow\{0,1\}$
a function that gets as input $l$ inputs for $f$ and returns $1$ if and only if the output of all the $l$ instance is $1$.
Formally,
$
f^{l}\left((x_{1}^{1},\dots,x_{k}^{1}),\dots,(x_{1}^{l},\dots,x_{k}^{l})\right)=f(x_{1}^{1},\dots,x_{k}^{1})\wedge f(x_{1}^{2},\dots,x_{k}^{2})\wedge\cdots\wedge f(x_{1}^{l},\dots,x_{k}^{l})
$.
The proof of the following lemma is straightforward,
albeit we attach the proof for completeness.

\begin{lemma}\label{Lem:FoolPump}
	Suppose $f$ has a $1$-fooling set of size $s$. Then $f^l$ has a $1$-fooling set of size $s^l$.
\end{lemma}

\begin{proof}
	Let $A=\{y^1=(x^1_1,\dots,x^1_k),\dots,y^s=(x^s_1,\dots,x^s_k)\}\subseteq  \{0,1\}^{n\times k}$ be a $1$-fooling set for $f$.
	We argue that $A^l$ ($l$-wise Cartesian product of $A$ with itself) is a $1$-fooling set for $\bigwedge f^l$.
	
	Indeed, for every $\left(y^{i_{1}},\dots,y^{i_{l}}\right)\in A^{l}$, it holds that
	\[
	 f^{l}\left(y^{i_{1}},\dots,y^{i_{l}}\right)=\bigwedge_{j=1}^{l}f(y^{i_{j}})=\bigwedge_{j=1}^{l}1=1~.
	\]
	Moreover, take two different points $(y^{i_{1}},\dots,y^{i_{l}})$ and $(y^{j_{1}},\dots,y^{j_{l}})$ in $A^{l}$.
	There is some index $r\in[l]$ such that $y^{i_{r}}\ne y^{j_{r}}$.
	Set $y^{i_{r}}=(z_{1},\dots,z_{k})$ and $y^{j_{r}}=(w_{1},\dots,w_{k})$. 
	As $y^{i_{r}},y^{j_{r}} \in A$, there is $x\in\left\{ z_{1},w_{1}\right\} \times\cdots\times\left\{ z_{k},w_{k}\right\} $ such that $f(x)=0$.
	In particular 
	\[
	f^{l}\left(y^{i_{1}},\dots,y^{i_{r-1}},x,y^{i_{r+1}},\dots,y^{i_{l}}\right)=\bigwedge_{j\in\left[l\right]\setminus\{r\}}f(y^{i_{j}})\wedge f(x)=\bigwedge_{\left[l\right]\setminus\{r\}}1\wedge0=0,
	\]
	as required.
\end{proof}

Now, we are ready to prove \Cref{lem:LB-NSD}.

\begin{proof}[Proof of \Cref{lem:LB-NSD}]
	Using \Cref{Lem:FoolPump} and \Cref{fact:HighFoolingSet},
	it will be enough to show that $\NSMf$ has a $1$-fooling set of size $\Omega(mk)$. 

	We start by defining $n$ metrics over $\{0,1\}^{m\times 2k}$:
	\[
	d_{i}\left(\left(A_{1},\dots,A_{2k}\right),\left(B_{1},\dots,B_{2k}\right)\right)=\left|\left\{ j\mid i\in A_{j}\bigtriangleup B_{j}\right\} \right|
	\]
	Here,
	$A_{j}\bigtriangleup B_{j}=\left(A_{j}\setminus B_{j}\right)\cup\left(B_{j}\setminus A_{j}\right)$ is the symmetric difference.\footnote{%
	  In fact $d_i$ is just the Hamming distance after we project the strings to the $i$'th coordinate.}
	\[
	d\left(\left(A_{1},\dots,A_{2k}\right),\left(B_{1},\dots,B_{2k}\right)\right)=\max_{i}d_{i}\left(\left(A_{1},\dots,A_{2k}\right),\left(B_{1},\dots,B_{2k}\right)\right)~.
	\]
	It is straightforward to verify that $d$ is indeed a metric. 
	
	Let $\mathcal{S}=\left\{ \left(A_{1},\dots,A_{2k}\right)\in\{0,1\}^{m\times 2k}\mid\forall i\in[m],~\left|\left\{ j\mid i\in A_{j}\right\} \right|=k\right\} $
	be all the points such that every index $i\in\left[m\right]$ appears
	in exactly $k$ sets. Note that $|\mathcal{S}|={2k \choose k}^{m}$,
    and that $\forall x\in \mathcal{S}$, $\NSMf(x)=1$.
	We will construct a subset $\mathcal{S}'\subseteq \mathcal{S}$ in a greedy manner.
	In each phase we will choose an arbitrary $x\in \mathcal{S}$, which was not deleted yet,
	add it to $\mathcal{S}'$ and delete all of $B(x,k/2)$, i.e.,
	all the points in $\mathcal{S}$ which are at distance at most $k/2$ from $x$ (with respect to the metric $d$).
	
	It holds that
	$$|B(x,k/2)\cap \mathcal{S}|=\left(\sum_{i=0}^{k/4}{k \choose i}^{2}\right)^{m}\le\left(2\cdot{k \choose k/4}^{2}\right)^{m}~.$$
	To see the equality,
	denote $x=\left(A_{1},\dots,A_{2k}\right)$.
	For each index $i\in [m]$, there are $k$ sets containing $i$.
	We should choose $j\le k/4$ sets to remove $i$ from, and $j$ new sets to insert $i$ into.
	All this is taken in power of $m$ as we have $m$ different indices.
	To see the inequity, note that for $i\le k/4$, ${k \choose i}/{k \choose i-1}=\frac{k-i+1}{i}>2$.
	Hence $\sum_{i=0}^{\frac{k}{4}-1}{k \choose i}^2<{k \choose k/4}^2$.

	By the end of the process {(when all the points in $\mathcal{S}$ were deleted)},
	we have a set $\mathcal{S}'$ of size at least $\left(\frac{{2k \choose k}}{2\cdot{k \choose k/4}^2}\right)^{m}$
	such that for every $x,y\in M'$, $d(x,y)\ge  k/2$. 
	We argue that $\mathcal{S}'$ is a $1$-fooling set.
	As $\mathcal{S}'\subseteq\mathcal{S}$, it holds that $\forall x\in\mathcal{S}'$, $\NSMf(x)=1$.
	Consider $x\ne y\in \mathcal{S}'$, where $x=\left(A_{1},\dots,A_{2k}\right)$ and $y=\left(B_{1},\dots,B_{2k}\right)$.
	There is an index $i$ such that $d_i(x,y)\ge k/2$.
	Therefore,
	$\left|\left\{ j\mid i\in A_{i}\cup B_{i}\right\} \right|\ge\frac{5}{4}k$.
	In particular,
	there is $ z\in\left\{ A_{1},B_{1}\right\} \times\cdots\times\left\{ A_{2k},B_{2k}\right\}$ with at least $\frac{5}{4}$ sets containing $i$,
	implying $\NSMf(z)=0$.
	
	Finally,
	we lower bound $|\mathcal{S}'|$. Recalling Stirling's formula,
	which says that $n!\approx\sqrt{2\pi n}\left(\frac{n}{e}\right)^{n}$,
	and the identity ${2k \choose k}=\sum_{i=0}^k{k\choose i}^2$,
	we have that:
	\[
	\frac{{2k \choose k}}{2{k \choose k/4}^{2}}\ge\frac{{k \choose k/2}^{2}}{{k \choose k/4}^{2}}=\frac{\left(\frac{k}{4}!\right)^{2}\left(\frac{3k}{4}!\right)^{2}}{\left(\frac{k}{2}!\right)^{4}}=\Omega(1)\cdot\frac{\left(\frac{k}{4}\right)^{\frac{1}{2}k}\left(\frac{3k}{4}\right)^{\frac{3}{2}k}}{\left(\frac{k}{2}\right)^{2k}}=\Omega(1)\cdot\left(\frac{3^{\frac{3}{2}}}{4}\right)^{k}~.
	\]
	We conclude that 
	\[
	D\left(\NSMf\right)\ge\log\left(\left|\mathcal{S}'\right|\right)\ge\log\left(\left(\Omega(1)\cdot\left(\frac{3^{\frac{3}{2}}}{4}\right)^{k}\right)^{m}\right)=\Omega(mk)~.\qedhere
	\]
\end{proof}

\end{document}